\documentclass[11pt]{article}
\usepackage[]{amsmath,amssymb,amsfonts,latexsym,amsthm,enumerate,fullpage}
\usepackage[latin9]{inputenc}
\usepackage{amsmath}
\usepackage{amssymb}
\usepackage{breqn}
\usepackage{url}
\usepackage{graphicx}
\usepackage[pdfstartview=FitH]{hyperref}
\usepackage{amsthm}
\usepackage{hyperref}
\usepackage{url}
\usepackage{enumitem}
\usepackage{mathtools}
\usepackage{verbatim}
\usepackage{cleveref}

\usepackage{bbm}
\newcommand{\Mat}{\operatorname{Mat}}

\newcommand{\Span}{\operatorname{span}}
\newcommand{\ground}{\operatorname{gr}}
\newcommand{\nint}{\operatorname{nint}}
\newcommand{\loc}{\operatorname{loc}}
\newcommand{\ld}{\operatorname{ld}}
\newcommand{\GL}{\operatorname{GL}}
\newcommand{\HDE}{\operatorname{HDE-System}}
\newcommand{\supp}{\operatorname{supp}}
\newcommand{\Aff}{\operatorname{Aff}}

\newcommand{\suppRow}{\operatorname{suppRow}}
\newcommand{\suppCol}{\operatorname{suppCol}}

\newcommand{\rej}{\operatorname{rej}}

\newcommand{\sph}{\operatorname{sph}}
\newcommand{\opp}{\operatorname{opp}}
\newcommand{\ext}{\operatorname{ext}}

\makeatletter
\begin{document}
\newtheorem{theorem}{Theorem}[section]
\newtheorem*{theorem*}{Theorem}
\newtheorem{lemma}[theorem]{Lemma}
\newtheorem{definition}[theorem]{Definition}
\newtheorem{claim}[theorem]{Claim}
\newtheorem{example}[theorem]{Example}
\newtheorem{remark}[theorem]{Remark}
\newtheorem{proposition}[theorem]{Proposition}
\newtheorem{corollary}[theorem]{Corollary}
\newtheorem{observation}[theorem]{Observation}

\title{High dimensional expansion implies amplified local testability}
\author{
Tali Kaufman
\footnote{Department of Computer Science, Bar-Ilan University, kaufmant@mit.edu, research supported by ERC and BSF.}
\and
Izhar Oppenheim
\footnote{Department of Mathematics, Ben-Gurion University of the Negev, Be'er Sheva 84105, Israel, izharo@bgu.ac.il, research supported by ISF (grant No. 293/18).}
}

\maketitle
\begin{abstract}
In this work we show that high dimensional expansion implies locally testable code. Specifically, we define a notion that we call high-dimensional-expanding-system (HDE-system). This is a set system defined by incidence relations with certain high dimensional expansion relations between its sets.  We say that a linear code is modelled over HDE-system, if the collection of linear constraints that the code satisfies could by described via the HDE-system. We show that a code that can be modelled over HDE-system is locally testable.

This implies that high dimensional expansion phenomenon {\em solely} implies local testability of codes. Prior work had to rely to local notions of local testability to get some global forms of testability  (e.g. co-systolic expansion from local one, global agreement from local one), while our work infers global testability directly from high dimensional expansion without relying on some local form of testability.

The local testability result that we obtain from HDE-systems is, in fact, stronger than standard one, and we term it {\em amplified local testability}. Roughly speaking in an amplified locally testable code, the rejection probability of a corrupted codeword, that is not too far from the code, is amplified by a $k$ factor compared to the guarantee in standard testing, where $k$ is the length of the test (the guarantee for a corrupted codeword that is very far from the code is the same as in standard testing). Amplified testing is stronger than standard local testability, but weaker than the notion of optimal testing as defined by Bhattacharyya et al., that roughly requires amplified local testability without assuming that the corrupted codeword is not too far from the code.

We further show that most of the well studied locally testable codes as Reed-Muller codes and more generally affine invariant codes with single-orbit property fall into our framework. Namely, it is possible to show that they are modelled over an HDE-system, and hence the family of all p-ary affine invariant codes is amplified locally testable. This yields the strongest known testing results for affine invariant codes with single orbit, strengthening the work of Kaufman and Sudan.

\end{abstract}

\section{Introduction}
\paragraph{High dimensional expansion implies testability.} The aim of this work is to show that codes arising from high dimensional expanding set systems are locally testable. Specifically, we define the notion of {\em High Dimensional Expanding System ($\HDE$)} that is a two layer expanding set system that generalizes two layer set systems arising from high dimensional expanders. Using this new concept, we show that codes whose constraints form an $\HDE$ are locally testable.

\paragraph{\bf Testability of well studied codes via high dimensional expansion.} We further show that most well studied locally testable codes as Reed-Muller codes and more generally affine-invariant codes are, in fact, $\HDE$ codes! Hence, their local testability could be re-inferred using our current work; and could be attributed to the high dimensional expansion phenomenon. Specifically, we give a high dimension expansion based proof to the local testability of single orbit affine invariant codes, that strengthen the well known result of Kaufman and Sudan \cite{KS}.

\paragraph{High dimensional expansion implies amplified testability.}
In the following we define locally testable codes, and a strictly stronger notion of testability that we term amplified locally testable codes. 
Roughly speaking in an amplified locally testable code, the rejection probability of a corrupted codeword, that is not too far from the code, is amplified by a $k$ factor compared to the guarantee in standard testing, where $k$ is the length of the test (the guarantee for a corrupted codeword that is very far from the code is the same as in standard testing). Our main goal would be to show that local testability could be inferred by high dimensional expansion machinery. Furthermore, in cases where high dimensional expansion implies locally testable codes, the derived codes are not only locally testable, but rather they are also amplified locally testable. By applying this machinery to single orbit affine invariant codes, we get that these codes are amplified locally testable, which is the strongest notion of testability known for such codes, strengthening the well known work of Kaufman and Sudan \cite{KS}.

\begin{definition} [Locally testable code]
Given a linear code $C \subseteq \mathbb{F}_p^V$ defined by a set $\mathcal{E}_C$ of $k$-query tests, define $\rej : \mathbb{F}_p^V \rightarrow [0,1]$ where $\rej (\underline{c})$ is the fraction of $k$-query tests that $\underline{c}$ fails (by definition,  $\underline{c} \in C$ if and only if $\rej (\underline{c}) = 0$).

Let $\mathcal{C}$ be a family of codes such that every $C \in \mathcal{C}$ is defined by a set $\mathcal{E}_C$ of $k (C)$-query tests.  We say that a family of linear codes $\mathcal{C}$ is  { \em  locally testable} if there are constants $t_{\mathcal{C}} \in \mathbb{N}, t_{\mathcal{C}} >0$ and $r_{\mathcal{C} } >0$ such that for every $C \in \mathcal{C}$ the following robustness property holds: For every $\underline{c} \in \mathbb{F}_p^V$,
$$\\rej (\underline{c}) \geq r_{\mathcal{C}} \min \left\lbrace \min_{\underline{c} ' \in C} \Vert \underline{c} - \underline{c} ' \Vert, \frac{1}{(k(C))^{t_{\mathcal{C}}}} \right\rbrace .$$
\end{definition}

\begin{remark}
We note that there are several variants to the definition of local testability. Here we adopt the one used by the work of Kaufman and Sudan on affine invariant codes (see for instance \cite[Theorem 2.9]{KS}).
\end{remark}



\begin{definition} [Amplified locally testable codes]
\label{amp. loc. test def.}
Let $\mathcal{C}$ be a family of codes such that every $C \in \mathcal{C}$ is defined by a set $\mathcal{E}_C$ of $k (C)$-query tests.  We say that a family of linear codes $\mathcal{C}$ is  { \em  locally testable} if there are constants $t_{\mathcal{C}} \in \mathbb{N}, t_{\mathcal{C}} >0$ and $r_{\mathcal{C} } >0$ such that for every $C \in \mathcal{C}$ the following robustness property holds: For every $\underline{c} \in \mathbb{F}_p^V$,
$$\\rej (\underline{c}) \geq k(C) r_{\mathcal{C}} \min \left\lbrace \min_{\underline{c} ' \in C} \Vert \underline{c} - \underline{c} ' \Vert, \frac{1}{(k(C))^{t_{\mathcal{C}}}} \right\rbrace .$$
\end{definition}

A few remarks are in order:
\begin{remark}[Amplified local testability is interesting when $k$ varies]
At first sight, the definition of amplified local testability might not seem very interesting. Indeed, if one considers a family of codes $\mathcal{C}$ in which the number of bits in the queries does not vary,  amplified local testability is the same as local testability, since $k = O(1)$. However, amplified local testability is meaningful when $k$ varies. Consider for example the family of all binary Reed-Muller codes, testing for polynomials in degree $d$ in $\mathbb{F}_2^{\mathbb{F}_2^n}$ where both $d$ and $n$ vary and $d <<n$ (e.g., codes where $100 d +100 \leq n$). For codes in this family $k=2^d$, but this is not constant in this family. I.e., $k$ is not $O(1)$ in this family and amplified local testability for this family is stronger than local testability.
\end{remark}


\begin{remark}[Role of $t_{\mathcal{C}}$]
The best we hope for amplified local testing is $t_{\mathcal{C}} = 1$. If this happens, then the family has optimal testing as defined in \cite{BKSSZ}.  The most famous example of optimal testing is the work of Bhattacharyya at el.  \cite{BKSSZ} that showed optimal testing for Reed-Muller codes. Our methods below do not give optimal testing, but only amplified testing with $t_{\mathcal{C}} = 3$.
\end{remark}




In this work we show that a code which can be described via $\HDE$, not only we can infer local testability for it, but rather we can infer amplified local testability for it!


\paragraph{Local testability via unique neighbor expansion.}
We show that $\lambda$-expanding $\HDE$ has some form of unique neighbor expansion property associated with it (see Definition \ref{uni-neigh-exp def}). We also show that if the HDE-system has a strong enough unique neighbor expansion property, then a linear code defined based on this system is amplified locally testable. We prove that this is the case for affine-invariant codes with the single orbit property. Thus, HDE-system provides a mechanism to get amplified local testability of codes. Prior to our work, there was not a general phenomenon that explains local testability. In this work we show that local testability, and in fact a stronger notion of amplified local testability, is a high dimensional expansion phenomenon!

\subsection{Comparison to prior works}
A recent work of Dikstein at el.   \cite{DDHRZ} seems superficially close to the methods of this paper, since both works deduce local testability using ideas stemming from high dimensional expansion. The reader should note that there are major difference between the works:
\begin{itemize}
\item The work of Dikstein at el.   \cite{DDHRZ} relies on the idea that "global" local testability can be inferred from "local" local testability.  I.e., in \cite{DDHRZ}, the assumption is that the code contains many small (i.e., "local") locally testable codes and by expansion considerations, it follows that the global code is locally testable. This is also the point of view of \cite{KKL,EK,KM} that considered what can be thought of as "co-cycle codes" and the global testability was be derived assuming they are composed of small local codes that are locally testable (aka "the links" code).  In contrast to \cite{DDHRZ} (and to \cite{KKL,EK,KM}), the focus of this current work is to get local testability of codes {\em directly} from high dimensional expansion phenomenon. Deducing local testability of codes directly from high dimensional expansion (without relying on any local code that is locally testable) is new and is achieved here for the first time.
\item Our work has the benefit of deducing not only local testability, but rather amplified local testability which was not achieved in \cite{DDHRZ}.
\item As far as we know, the work of \cite{DDHRZ} does not apply to the family of affine invariant codes, but only to a sub-family of lifted codes. Thus, in terms of generality, our work seems to apply in a more general setting.
\end{itemize}

It is also beneficial to compare the results of this paper to previous results regarding single orbit affine invariant codes. In \cite[Theorem 2.9]{KS}, it was shown that single orbit affine invariant codes are locally testable.  Using our new machinery,  we improve on this result, showing the the family of all affine invariant codes has amplified local testability. As noted above, a stronger result was known for Reed-Muller codes (which is a sub-family of the family of affine invariant codes), but, prior to our work, no general treatment was available to the entire family of single orbit affine invariant codes.

\subsection{High Dimensional Expanding System ($\HDE$)}
Our main definition towards defining high dimensional expander codes is called High-Dimensional-Expanding-System or $\HDE$ for short.

We start by defining a $(s,k,K)$-Two layer system:
\begin{definition}[$(s,k,K)$-Two layer system]
\label{kK-T-L-S def}
A two layer system $X$ is a system $X= (V,E,T)$ of three sets:
\begin{enumerate}
\item A finite set $V$ whose elements are called vertices.
\item A set $E \subseteq 2^V$ such that $\vert \tau \vert = k$ for every $\tau \in E$ and $\bigcup_{\tau \in E} \tau = V$.
\item A set $T \subseteq 2^E$ such that $\vert \sigma \vert = K$ for every $\sigma \in T$ and $\bigcup_{\sigma \in T} \sigma = E$.
\item By abuse of notation, we will denote $v \in \sigma$ for $v \in V, \sigma \in T$ if there is $\tau \in \sigma$ such that $v \in \tau$. Using this notation, for every $\sigma \in T$ and every $v \in \sigma$,
$$2 \leq \vert \lbrace \tau \in \sigma : v \in \tau \rbrace \vert \leq s.$$
\end{enumerate}
\end{definition}

Roughly speaking, $\HDE$ is a two layer system with good expansion properties. In order to give the definition, we need to define several graphs associated with a two layer system. We note that all the graphs defined below will be actually considered as weighted graphs with a weight function induced by weights on $T$, but in the introduction we suppress this fact in order to keep things simple.
\begin{definition}[The ground graph]
For a two layer system $X= (V,E,T)$, the ground graph of $X$ is the graph whose vertices are $V$ and edges are $\lbrace \lbrace v,u \rbrace : \exists \tau \in E, u,v \in \tau \rbrace$.
\end{definition}

\begin{definition}[Link of a vertex]
For a two layer system $X= (V,E,T)$ and $v \in X$, the link of $v$ is the graph whose vertex set is $E_v = \lbrace \tau \in E : v \in \tau \rbrace$ and  whose edge set is
$$T_v = \lbrace \lbrace \tau, \tau ' \rbrace : \tau \neq \tau ' \text{ and } \exists \sigma \in T \text{ such that } \tau, \tau ' \in \sigma \rbrace.$$
\end{definition}

\begin{definition}[The non-intersecting graph]
For a two layer system $X= (V,E,T)$, the non-intersecting graph of $X$ is a graph whose vertex set is $E$ and edge set is
$$\lbrace \lbrace \tau, \tau ' \rbrace : \tau \cap \tau ' = \emptyset \text { and } \exists \sigma \in T, \text{ such that } \tau, \tau ' \in \sigma \rbrace.$$
This graph corresponds to the {\em Non-Intersecting Walk}, i.e., to the walk from a between elements of $E$ that {\em do NOT intersect} (as subsets of $V$) via a $T$ element that contains both of them.
\end{definition}

An $\HDE$ is a two layer system $X$ in which all these graphs are expanding. More precisely, for $0 \leq \lambda <1$, we call a (weighted) graph $G$ a $\lambda$-expander is it is connected and either the second largest eigenvalue of the is $\leq \lambda$ or (which is less restrictive) its (generalized) Cheeger constant is $\geq 1- \lambda$ (see Definition \ref{Cheeger constant def} below).

\begin{definition} (High Dimensional Expanding System ($\HDE$)) [informal, for formal see Definition \ref{HDE formal def}]
For $0 \leq \lambda <1$, a (weighted) two layer system $X=(V,E,T)$ is called $\lambda$-expanding-$\HDE$ if the ground graph and the links of all the vertices are $\lambda$-expanders and the non-intersecting graph is either totally disconnect (i.e., it has no edges) or a $\lambda$-expander.
\end{definition}


\paragraph{High Dimensional expanders imply $\HDE$.} Part of our motivation for the Definition of HDE-systems is to mimic the definition of high dimensional expanders based on simplicial complexes (called $\lambda$-local spectral expander - see \cite[Definitions 2,3]{KO-Random}). The simplest example is when $Y$ is a $2$-dimensional simplicial complex. In this case, we define a two layer system $X=(V,E,T)$ as follows: $V$ is the vertex set of $Y$, $E$ is the edge set of $Y$ and $T$ is the sets of triples of edges that form a triangle in $Y$. We note that in this case the parameters of $X$ are $s =2, k = 2, K = 3$. Note that the ground graph is the $1$-skeleton of $Y$, the link of each vertex in $X$ is the link in the simplicial complex and the non-intersecting graph is totally disconnected (since every two edges that are in the same triangle share a vertex). Thus, by definition if $Y$ is a $\lambda$-local spectral expander, then the $1$-skeleton of $Y$ are all the links of $Y$ are $\lambda$-expanders and it follows that $X$ is $\lambda$-expanding.

\paragraph{Expanding HDE's  have unique neighbor expansion for small sets that are also locally small.} Our main motivation for the definition of $\HDE$ is the ability to deduce unique neighbor expansion theorem from them, for "small" sets that are also "locally small". This unique neighbor expansion theorem that we state below will play a major role in proving local testability based on $\HDE$.

In order to state this Theorem, we will need the following definition:

\begin{definition}($\delta$-Locally-small set) [informal, for formal see Definition \ref{locally small def}]
Let $X = (V,E,T)$ be a two layer system and let $A \subseteq E$ be a non-empty set. For a vertex $v \in V$, define $A_v = \lbrace \tau \in A : v \in \tau \rbrace$. For a constant $0 \leq \delta < 1$, a vertex $v$ is called $\delta$-small if the size of $A_v$ in the link of $v$ (when accounting for the weight function on the link) is smaller than $\delta$ fraction of the size of $E_v$. Vertices that are not $\delta$-small are called $\delta$-large. A set $A \subseteq E$ is called {\em $\delta$-locally small}, if the fraction of its mass that is distributed on vertices that are $\delta$-large is negligible with respect to the total mass of $A$. 
\end{definition}

Following we define a notion of unique neighbor expansion applies for small sets that are also $\delta$-locally small. 

\begin{definition}(Unique neighbor expansion property), [informal, for formal see Definition \ref{uni-neigh-exp def}]
We say that $A \subset E$ has a unique neighbor into $T$ if there exists $\sigma \in T$ that contains exactly one $k$-set from $A$.
Let $X = (V,E,T)$ be a two layer system and let $A \subseteq E$ be a non-empty set.
For constants  $\varepsilon_0 > 0$ , $\delta > 0$, we say that $X$ has {\em $(\delta, \varepsilon_0)$-unique neighbor expansion property} if for every non-empty set $A \subseteq E$ and every $\varepsilon < \varepsilon_0$ if $A$ $\varepsilon$-small (i.e., its mass is at most a $\varepsilon$-fraction of the total mass of $E$) and $\delta$-locally small, then $A$ has unique neighbor expansion into $T$.
\end{definition}

\begin{theorem}(Main Theorem 1: Unique neighbor expansion property for $\HDE$) [informal, for formal see Theorems
\ref{main exp thm - detailed},  \ref{HDE implies unique neigh exp thm}]
\label{main thm 1 - intro}
Given a $\lambda$-expanding $\HDE$ $X$, with $\lambda$ sufficiently small, there are $\delta >0$ and $\varepsilon_0 >0$ such that $X$ has the $(\delta, \varepsilon_0)$-unique neighbor expansion property. Moreover, if $s=2$, then $\delta \rightarrow 1$ as $\lambda \rightarrow 0$.
%

\end{theorem}

\paragraph{On the ability to get unique neighbor expansion from HDE-systems.}
The idea behind the proof of Main Theorem 1 is to use the expansion of the links in order to derive unique neighbor expansion.  The links are very good expanders so a set that is locally small has the property that its local views in the links expand a lot. Each link induces by its local view many "potential unique neighbors". However, it could be that the local views of the links will interfere and the "potential unique neighbors" by the "links opinion" will turn out to be non unique neighbors. Since the system is expanding the total interference between links is small and thus the overall unique neighbor property is implied.

\subsection{$\HDE$ Codes}

Given a two layer system $X = (V,E,T)$ as above, we want to use it as a "foundation" and for constructing a code. Such a construction it is not unique and cannot be done for every $X$. However, for a code that "could be constructed via $X$", its testability could be inferred from the expansion properties of $X$.

Before describing this construction, we need to establish some terminology and notation: Let $C \subseteq \mathbb{F}_p^V$ be a linear code (where $V$ is a finite set) with a check matrix $H$.
\begin{itemize}
\item We denote by $\mathcal{E} = \mathcal{E} (H)$ the rows $H$ and we refer to $\mathcal{E}$ as the constraints of the code (or $k$-constraints if they all have a support of size $k$ - see below). Thus, $\mathcal{E}$ are $1 \times n$ vectors and for $\underline{c} \in \mathbb{F}_p^V$, $\underline{c} \in C$ if and only if for every $\underline{e} \in \mathcal{E}$, $\underline{e} \cdot \underline{c} = 0$ (recall that $\underline{e}, \underline{c}$ are indexed by the elements in $V$, thus $\underline{e} \cdot \underline{c} = \sum_{v} \underline{e} (v) \underline{c} (v)$).
\item For $\underline{e} \in \mathcal{E}$, we define the support of $\underline{e}$ as
$$\supp (\underline{e}) = \lbrace v \in V : \underline{e} (v) \neq 0 \rbrace.$$
\item A \textit{linear dependency} of $\mathcal{E}$ is a function $\ld : \mathcal{E} \rightarrow \mathbb{F}_p$ such that for every $\underline{c} \in \mathbb{F}_p^V$, $\sum_{\underline{e} \in \mathcal{E}} \ld (\underline{e}) (\underline{e} \cdot \underline{c})  = 0$. In other words, if we think of the row vector $\underline{\ld} = (\ld (\underline{e}))_{\underline{e} \in \mathcal{E}}$, then $\underline{\ld} H = \underline{0}$. As above, the support of $\ld$ is the set
$$\supp (\ld) = \lbrace \underline{e} \in \mathcal{E} : \ld (\underline{e}) \neq 0 \rbrace.$$
\end{itemize}

\begin{example}
\label{example of a code with ld}
Consider $C \subseteq \mathbb{F}_2^V$, $V = \lbrace v_1, v_2 \rbrace$ given by the parity check matrix
$$H = \begin{pmatrix}
1 & 0  \\
0 & 1  \\
1 & 1
\end{pmatrix}.$$
If $\underline{e}_i$ denotes the $i$-th row of $H$, then $\ld : \lbrace \underline{e}_1, \underline{e}_2, \underline{e}_3 \rbrace \rightarrow \mathbb{F}_2$ defined by
$$\ld (\underline{e}_i) = 1, \forall i=1,2,3,$$
is a linear dependency. Indeed,
$$\underline{\ld} = \begin{pmatrix}
1 & 1 & 1
\end{pmatrix},$$
and one can verify that $\underline{\ld} H = \underline{0}$.
\end{example}

\begin{definition}[Code modelled over a two layer system]
\label{code modelled over t-l-s def}
Let $X = (V,E,T)$ be a two layer system. A code $C$ is said to be modelled over $X$ if the following holds:
\begin{itemize}
\item There is a prime power $p$ such that $C \subseteq \mathbb{F}_p^V$.
\item There is a check matrix $H$ and $\mathcal{E} = \mathcal{E} (H)$ such that
$$E= \lbrace \supp (\underline{e}) : \underline{e} \in \mathcal{E} \rbrace,$$
and such that for every  $\underline{e}_1, \underline{e}_2 \in \mathcal{E}$, if $\underline{e}_1 \neq \underline{e}_2$, then $\supp (\underline{e}_1) \neq \supp (\underline{e}_2)$. In other words, there is a bijection $\Phi : \mathcal{E} \rightarrow E$ given by $\Phi (\underline{e}) = \supp (\underline{e})$. Note that under this assumption, the size of the support of all the constraints in $k$ (the constant of the system $X$) and we refer to the elements of $\mathcal{E}$ as the $k$-constraints of the code, when there is no chance for ambiguity. 
\item There is a set $\mathcal{T}$ of linear dependencies such that
$$T= \lbrace \lbrace \supp (\underline{e}) : \underline{e} \in \supp (\ld) \rbrace : \ld \in \mathcal{T} \rbrace.$$
\end{itemize}
\end{definition}

\begin{example}
Let $X = (V,E,T)$ the following two layer system: $V = \lbrace v_1, v_2, v_3 \rbrace$,
$E = \lbrace \tau_{i,j} = \lbrace v_i, v_j \rbrace : 1 \leq i < j \leq 3 \rbrace$ and $T = \lbrace \sigma = \lbrace \tau_{i,j} : 1 \leq i < j \leq 3 \rbrace \rbrace$. Then for every prime power $p$, we can define a code $C \subseteq \mathbb{F}_p^V$ modelled over $X$ as follows: define the check matrix of the code to be
$$H = \begin{pmatrix}
1 & p-1 & 0 \\
0 & 1 & p-1 \\
p-1 & 0 & 1
\end{pmatrix}.$$
One can see that for this matrix the support of the $i$-the row is $\lbrace v_i, v_{i+1 \mod 3} \rbrace \in E$ and that no two rows have the same support. Further define a linear dependency $\ld : \mathcal{E} \rightarrow \mathbb{F}_p$ to be the constant function $1$, thus one can verify that the support of $\ld$ is $\sigma \in T$ and that this is indeed a linear dependency.
\end{example}

Our motivation for considering codes modelled over two layer system is the following

\begin{theorem}(Main Theorem 2: Codes modelled over two layer systems with unique neighbor property are amplified locally testable)  [informal, for formal see Corollary \ref{unique neighbor exp imply amp. loc. test. coro}] For every $p$ prime, $t' \in \mathbb{N}, t' >0$, $\mu >0$ and $\delta > \frac{p-1}{p}$, let $\mathcal{C} (\delta, p, t', \mu)$ be the family of $p$-ary codes (i.e.., codes of the form $C \subseteq \mathbb{F}_p^{V (C)}$ ) modelled over two layer systems such that
\begin{dmath*}
\mathcal{C} (\delta, p, t' ,\mu) =
{\left\lbrace C : \exists \varepsilon_0 (C) >0 \text{ such that } C \text{ has the } (\delta, \varepsilon_0 (C)) \text{-unique neighbor property and } \varepsilon_0 \geq \frac{\mu}{k(C)^{t'} } \right\rbrace.}
\end{dmath*}
Then the family $\mathcal{C} (\delta, p, t', \mu)$ is amplified locally testable with $t_{\mathcal{C}} = t'+1$.   
\end{theorem}

\paragraph{On the ability to get local testability from unique neighbor expansion.} We assume we are in a situation that we have a code that is modelled over an HDE-system. Thus, we know that each $k$-constraint of the code is participating in a linear dependency.
This means that on every dependency, if there is one violated constraint that touches it, there must be another one that touches it.

We are given a vector that falsifies few constraints from the code and we want to show that such a vector is close to the code.
We can try to correct it by flipping variables such that this flipping reduces the number of violated constraints.
Clearly, such a procedure will not flip more variables then the total number of unsatisfied constraints which we assume to be small.

If you arrived at a situation when no more flipping can be done then you know that on each variable (assuming we work over $\mathbb{F}_p$) the fraction of violating constraints is at most $\frac{p-1}{p}$ of the constraints that it participates in (making it $\frac{p-1}{p}$-locally small). However, the unique neighbor expansion implies that there are linear dependencies that "sees" only one violating constraint! but as we said, this is not possible. So when we arrived at a situation where no more flipping is possible, we, in fact, arrived at a codeword that is close to our initial vector as required.


\begin{definition} ($\HDE$ Code)
We call a code $C$ as above a $\HDE$-code if it is modelled over a $\lambda$-expanding $\HDE$ system.
\end{definition}

\paragraph{Codes that give rise to $\HDE$ with $s =2$ are amplified locally testable.}

By Main Theorem 2, the family $\mathcal{C}_{\delta, p}$ of codes $C \subseteq \mathbb{F}_p^V$ modelled a two layer systems with a  $(\delta, \varepsilon_0 (C))$-unique neighbor property are locally testable given that $\delta > \frac{p-1}{p}$. We have furthered showed (see Main Theorem 1 above) that given any $\delta <1$, there is $\lambda$ sufficiently small such that every $\lambda$-expanding $\HDE$ with $s=2$ has the $(\delta, \varepsilon_0)$-unique neighbor property (where $\varepsilon_0$ depends on the parameters of the $\HDE$).
Thus, overall we get that the family of all codes $C \subseteq \mathbb{F}_p^V$ modelled over $\lambda$-expanding $\HDE$ with $s=2$ (and $\lambda$ sufficiently small) is amplified locally testable.

\begin{corollary}(Codes modelled over expanding-$\HDE$ with $s=2$ are amplified locally testable)  [informal, for formal see Theorem \ref{main loc test thm}] The family of all codes $C \subseteq \mathbb{F}_p^V$ of $k$-constraints modelled over expanding HDE systems with $s=2$ is amplified locally testable. Moreover, under some mild assumptions (passing to a large sub-family) $t_{\mathcal{C}} =3$ where $t_{\mathcal{C}}$ is as in Definition \ref{amp. loc. test def.}.
\end{corollary}



\paragraph{Local testability of single orbit affine invariant codes via $\HDE$.}
In the following we refer to single orbit affine invariant codes which were shown to be locally testable by the Kaufman-Sudan work \cite{KS}. These codes contain the well known Reed-Muller codes. We show that they are $\HDE$ codes with $s=2$, so their local testability is implied by our current work.
Kaufman and Sudan have shown that single orbit affine invariant codes which are characterized by $k$-weight constraints that form a single orbit are locally testable. We will show that the Kaufman-Sudan requirement allows to show that single orbit affine invariant codes are modelled over $\HDE$ with $s=2$ and thus are amplified locally testable.

\begin{theorem}(Testability of single orbit affine invariant codes) [informal, for formal, see Theorem \ref{affine inv local test gen. thm} and Corollary \ref{affine inv local test fixed delta coro}]
Let $\mathcal{C}_{\text{affine-inv,p}}$ be the family of all single orbit affine invariant codes $C \subseteq \mathbb{F}_p^{\mathbb{K (C)}^{n (C)}}$ with
$$\vert \mathbb{K} (C) \vert^{n (C)} \geq 2^{11} p^2 (k (C))^4,$$
where $k(C)$ is the size of the support of the constraint defining $C$. Then these codes are modelled over $\HDE$ with $s=2$ and hence the family of all these codes is amplified locally testable via our machinery.  Explicitly, for every $C \in \mathcal{C}_{\text{affine-inv}, p}$ and every $\underline{c} \in \mathbb{F}_p^{\mathbb{K} (C)^{n (C)}}$ it holds that
$$\rej (\underline{c}) \geq k (C) \frac{1}{2^{15} p^4} \min \left\lbrace \min_{\underline{c} ' \in C} \Vert \underline{c} - \underline{c} ' \Vert, \frac{1}{k(C)^3} \right\rbrace.  $$
\end{theorem}

We compare this result to the (non-amplified) local testing for affine invariant codes of Kaufman and Sudan \cite[Theorem 2.9]{KS} who showed the following:
\begin{theorem}\cite[Theorem 2.9]{KS}
For every $C \in \mathcal{C}_{\text{affine-inv,p}}$ it holds that
$$\rej (\underline{c}) \geq \frac{1}{2} \min  \left\lbrace \min_{\underline{c} ' \in C} \Vert \underline{c} - \underline{c} ' \Vert, \frac{1}{k(C)^2} \right\rbrace.  $$
\end{theorem}

Our Theorem and \cite[Theorem 2.9]{KS} both give a rejection of $\Omega (\frac{1}{k (C)^2})$ when $\min_{\underline{c} ' \in C} \Vert \underline{c} - \underline{c} ' \Vert$ is large. However, when $\min_{\underline{c} ' \in C} \Vert \underline{c} - \underline{c} ' \Vert << \frac{1}{k (C)^3}$,  and $k(C)$ is large, our result gives a much better rejection rate.

\paragraph{Local testability when $s>2$.}
The main focus of this work is proving amplified local testability for codes modelled over HDE-systems with $s=2$. We further have a more general treatment for codes modelled over HDE-system with general $s \geq 3$ under some extra-assumptions. Roughly speaking, for the case of $s \geq 3$ we need the extra assumption that the code is composed of local small codes that are locally testable. The difficulty in the case where $s \geq 3$ is that the bit flipping argument we described above can only correct the code to be $\frac{p-1}{p}$-locally small, while "The unique neighbor Theorem" says that we can deduce the $(\delta, \varepsilon_0)$-unique neighbor property from expansion given that $\delta < \frac{1}{s-1}$. Thus, in the case where $s \geq 3$, we may not be able correct a corrupted codeword by bit flipping to a setting in which we can apply our unique neighbor argument. This difficulty is dealt by adding the assumption of "local" local testability that grantees that correcting by bit flipping converges to a word that is $\delta$-locally small (and thus we can use our previous machinery). This new method requires some additional definitions and we refer the reader to Section \ref{Sphere corrections for general HDE codes} for further details.

\paragraph{Distance of HDE codes.}
An additional result is that for codes modelled over HDE-systems, the distance of the code can be bounded in terms of the expansion of the HDE system. This is done explicitly in Section \ref{Distance of codes modelled over two layer systems sec} (see Theorem \ref{distance thm}).

\paragraph{Organization of this paper:} In Section \ref{Two layer set systems sec}, we define weighted two layer systems and the induced weight functions on the graphs associated to it (the ground graph, the non-intersecting graph and the links). In Section \ref{Codes modelled over two layer systems and local testability sec}, we prove Main Theorem 2 above, showing how unique neighbor expansion implies (amplified) local testability. In Section \ref{Main Expansion Theorem sec}, we state and prove Main Theorem 1 above. showing that an expanding $\HDE$ has the unique neighbor property. In Section \ref{Application of the main expansion Theorem to local testability sec}, we show how to combine both our main Theorems in order to deduce amplified local testability for codes modelled over HDE-systems.  In Section \ref{Example: Single orbit affine invariant codes sec}, we show that our new machinery can be used to show that the family of p-ary single orbit affine invariant codes is amplified locally testable.  In Section \ref{Sphere corrections for general HDE codes}, we explain how to deduce local testability for codes modelled over HDE-systems with $s \geq 3$ under additional assumptions of "local" local testability. In Section \ref{Distance of codes modelled over two layer systems sec}, we show that codes modelled over $\HDE$ also have a bound on the distance of the code coming from expansion considerations.

\section{Two layer set systems}
\label{Two layer set systems sec}
Let $X = (V,E,T)$ be a two layer system (see Definition \ref{kK-T-L-S def}). We fix a function $w: T \rightarrow \mathbb{R}_{+}$ and define the weight function $w: V \cup E \cup T \rightarrow \mathbb{N}$ as follows:
$$\forall \tau \in E, w(\tau) = \sum_{\tau \in \sigma} w (\sigma),$$
$$\forall v \in V, w(v) = \sum_{v \in \sigma} w (\sigma).$$
Also, for a set $A \subseteq V \cup E \cup T$, we define
$$w(A) = \sum_{\eta \in A} w (\eta).$$

\begin{proposition}
\label{w(v) as sum of w(tau)'s prop}
For every $v \in V$,
$$2 w(v) \leq w (\lbrace \tau \in E : v \in \tau \rbrace) \leq s w(v).$$
\end{proposition}

\begin{proof}[Proof of Proposition \ref{w(v) as sum of w(tau)'s prop}]
Fix $v \in V$. By definition of $w$ it holds that
\begin{dmath*}
w (\lbrace \tau \in E : v \in \tau \rbrace) = \sum_{\tau \in E, v \in \tau} w (\tau) =
\sum_{\tau \in E, v \in \tau} \left( \sum_{\sigma \in T, \tau \in \sigma} m(\sigma) \right) =
\sum_{\sigma \in T, v \in \sigma} m(\sigma) \left( \sum_{\tau \in E, v \in \tau \subseteq \sigma} 1 \right) \geq \sum_{\sigma \in T, v \in \sigma} 2 m(\sigma) = 2 w(v),
\end{dmath*}
and the proof of the second inequality is similar, using the fact that
$$\sum_{\tau \in E, v \in \tau \subseteq \sigma} 1 \leq s.$$
\end{proof}

\begin{corollary}
\label{w(V) and w(E) coro}
For every $X$ as above,
$$ \frac{2}{k} w(V) \leq w(E) \leq  \frac{s}{k} w(V).$$
\end{corollary}

\begin{proof}[Proof of Corollary \ref{w(V) and w(E) coro}]
By Proposition \ref{w(v) as sum of w(tau)'s prop},
\begin{dmath*}
w(V) = \sum_{v \in V} w (v) \leq
\sum_{v \in V} \frac{1}{2} \sum_{\tau \in E, v \in \tau} w(\tau) =
\frac{1}{2} \sum_{\tau \in E} w(\tau) \sum_{v \in \tau} 1 =
\frac{k}{2} w (E),
\end{dmath*}
and the proof of the second inequality is similar.
\end{proof}

\subsection{Weighted graphs}

For a finite graph $G=(V,E)$, a \textit{weight function} on the edges is any positive function $m: E \rightarrow \mathbb{R}^+$. This induces a positive function $m : V \rightarrow \mathbb{R}^+$ by
$$\forall v \in V, m (v) = \sum_{e \in E, v \in e} m(e).$$
Also, for $\emptyset \neq U \subseteq V$, we denote
$$m (U) = \sum_{v \in U} m (v) .$$
A \textit{weighted graph} is a graph $(V,E)$ a weight function.

\subsection{The ground graph}

Recall that the ground graph of $X$ is the $1$-skeleton of the hyper-graph $(V,E)$. Concretely, the ground graph is the graph with the vertex set $V$ such that for every $u,v \in V$, $v$ and $u$ are connected by an edge if there is $\tau \in E$ such that $v,u \in \tau$. Define a weight function $m_{\ground}$ on this graph as
$$m_{\ground} (\lbrace v,u \rbrace) = \sum_{\tau \in E, v,u \in \tau} w (\tau).$$

\begin{proposition}
\label{m_gr (v) prop}
For every $v \in V$,
$2(k-1) w (v) \leq m_{\ground} (v) \leq s(k-1) w (v)$.
\end{proposition}

\begin{proof}[Proof of Proposition \ref{m_gr (v) prop}]
Fix $v \in V$. By the definition of $m_\ground$,
\begin{dmath*}
m_{\ground} (v) = \sum_{u \in V, u \neq v} m_{\ground} (\lbrace v,u \rbrace) =
\sum_{u, u \neq v} \left( \sum_{\tau \in E, v,u \in \tau} \left( \sum_{\sigma \in T, \tau \in \sigma} w (\sigma) \right) \right) =
\sum_{\sigma \in T, v \in \sigma} w (\sigma) \left( \sum_{\tau \in E, v \in \tau \in \sigma} \left( \sum_{u, u \neq v, u \in \tau} 1 \right) \right)=
\sum_{\sigma \in T, v \in \sigma} w (\sigma) \left( \sum_{\tau \in E, v \in \tau \in \sigma} k-1 \right).
\end{dmath*}
Recall that for every $\sigma \in T$ and every $v \in \sigma$ it holds that
$$2 \leq \vert \lbrace \tau \in E : v \in \tau \in \sigma \rbrace \vert \leq s$$
and thus
$$\sum_{\sigma \in T, v \in \sigma} w (\sigma) \left( \sum_{\tau \in E, v \in \tau \in \sigma} k-1 \right) \geq \sum_{\sigma \in T, v \in \sigma} 2(k-1) = 2(k-1) w(v),$$
$$\sum_{\sigma \in T, v \in \sigma}   w (\sigma) \left( \sum_{\tau \in E, v \in \tau \in \sigma} k-1 \right) \leq \sum_{\sigma \in T, v \in \sigma} s(k-1) = s(k-1) w(v).$$
\end{proof}

\begin{proposition}
\label{m-ground (U,U) prop}
Let $U \subseteq V$ be a non-empty set. For $1 \leq i \leq k$, denote
$$E^i_U = \lbrace \tau \in E : \vert \tau \cap U \vert =i \rbrace.$$
Then
$$m_{\ground} (U,U) = \sum_{i =2}^k i (i-1) w(E^i_U).$$
\end{proposition}

\begin{proof}[Proof of Proposition \ref{m-ground (U,U) prop}]
By direct computation,
\begin{dmath*}
 m_\ground (U,U) =
\sum_{u \in U} \left( \sum_{v \in U, u \neq v} m_\ground (\lbrace u,v \rbrace) \right) =
\sum_{u \in U} \left( \sum_{v \in U, u \neq v} \left( \sum_{\tau \in E, u,v \in \tau} w (\tau) \right) \right) =
\sum_{u \in U} \left( \sum_{v \in U, u \neq v} \left( \sum_{\tau \in E, u,v \in \tau} w (\tau) \right) \right) =
\sum_{\tau \in E} w (\tau) \left( \sum_{u \in U, u \in \tau} \left( \sum_{v \in U, u \neq v, v \in \tau} 1 \right) \right) =
\sum_{i=2}^k \sum_{\tau \in E^i_U} w (\tau) \left( \sum_{u \in U, u \in \tau} \left( \sum_{v \in U, u \neq v, v \in \tau} 1 \right) \right) =
\sum_{i=2}^k \sum_{\tau \in E^i_U} w (\tau) i(i-1) =
\sum_{i =2}^k i (i-1) w(E^i_U).
\end{dmath*}
\end{proof}

\subsection{The non-intersecting graph}

Given $X$ as above, $\tau, \tau ' \in E$ will be called \textit{non-intersecting} and we will denote $\tau \leftrightarrow \tau'$ if $\tau \cap \tau ' = \emptyset$ and there is $\sigma \in T$ with $\tau, \tau' \in \sigma$.

We recall that the non-intersecting graph of $X$ is the graph with the vertex set $E$ and two elements $\tau, \tau ' \in E$ are connected by an edge if they are non-intersecting. Define a weight function on the edges of this graph as
$$m_{\nint} (\lbrace \tau, \tau ' \rbrace) = \sum_{\sigma \in T, \tau, \tau ' \in \sigma} w(\sigma).$$

Define the constants $Q_{\nint}^\max, Q_{\nint}^\min$ of $X$ as
$$Q_{\nint}^\max = \max_{\tau \in E} \max_{\sigma \in T, \tau \in \sigma} \vert \lbrace \tau ' \in E : \tau ' \in \sigma \text{ and } \tau \leftrightarrow \tau'  \rbrace \vert,$$
$$Q_{\nint}^\min = \min_{\tau \in E} \min_{\sigma \in T, \tau \in \sigma} \vert \lbrace \tau ' \in E : \tau ' \in \sigma \text{ and } \tau \leftrightarrow \tau'  \rbrace \vert.$$

\begin{proposition}
\label{Q opp prop}
For every $A \subseteq E$,
$$Q_\nint^\min  w (A) \leq m_\nint (A) \leq Q_\nint^\max w (A).$$
\end{proposition}

\begin{proof}[Proof of Proposition \ref{Q opp prop}]
It is enough to prove that for every $\tau \in E$, $Q_\nint^\min w (\tau) \leq m_\nint (\tau) \leq  Q_\nint^\max w (\tau)$.
Fix $\tau \in E$, then
\begin{dmath*}
m_\nint (\tau) = \sum_{\tau ', \tau \leftrightarrow \tau '} m_{\nint} ( \lbrace \tau, \tau ' \rbrace) =
\sum_{\tau ', \tau \leftrightarrow \tau '} \sum_{\sigma \in T, \tau, \tau ' \in \sigma} w(\sigma)  =
\sum_{\sigma \in T, \tau  \in \sigma} w(\sigma) \sum_{\tau ', \tau \leftrightarrow \tau ', \tau ' \in \sigma} 1
\end{dmath*}
and the inequalities follow from the fact that
$$Q_\nint^\min \leq \sum_{\tau ', \tau \leftrightarrow \tau ', \tau ' \in \sigma} 1 \leq Q_\nint^\max.$$
\end{proof}

We define the \textit{regularity constant} of the non-intersecting graph to be
$R_\nint = \frac{Q_{\nint}^\min}{Q_{\nint}^\max}$.

\begin{remark}
Note that $0 \leq R_\nint \leq 1$ and $R_\nint >0$ if and only if the non-intersecting graph has no isolated vertices.
\end{remark}
The following Corollary readily follows from Proposition \ref{Q opp prop}:

\begin{corollary}
\label{frac of m_opp coro}
Let $X$ as above and assume that $R_\nint >0$. Then for every two non-empty sets $A, B \subseteq E$,
$$\frac{1}{R_\nint} \frac{w(A)}{w(B)} \geq \frac{m_\nint (A)}{m_\nint (B)}.$$
\end{corollary}

\subsection{Links of vertices}

Recall that for every $v \in V$, the link graph of $v$, denoted $X_v$, is the graph with the vertex set $E_v = \lbrace \tau \in E : v \in \tau \rbrace$ and the edge set:
$$T_v = \lbrace \lbrace \tau_1, \tau_2 \rbrace : \tau_1 \neq \tau_2 \text{ and } \exists \sigma \in T, \tau_1, \tau_2 \in \sigma \rbrace.$$
Define a weight function on $X_v$ as follows: for every $\lbrace \tau_1, \tau_2 \rbrace \in T_v$, define
$$m_v (\lbrace \tau_1, \tau_2 \rbrace) = \sum_{\sigma \in T, \tau_1, \tau_2 \in \sigma} w(\sigma).$$

\begin{proposition}
\label{m_v (tau), m_v (X_v) prop}
For every $v \in V$ and every $\tau \in E_v$,
$$w (\tau) \leq m_v (\tau) \leq (s-1) w (\tau),$$
and,
$$2 w(v) \leq m_v (E_v) \leq s w(v).$$
\end{proposition}

\begin{proof}[Proof of Proposition \ref{m_v (tau), m_v (X_v) prop}]
Fix $v \in V$ and $\tau \in E_v$, then
\begin{dmath*}
m_v (\tau) = \sum_{\tau ' \in E_v, \tau ' \neq \tau} m_v (\lbrace \tau, \tau ' \rbrace) =
\sum_{\tau' \in E, v \in \tau ', \tau ' \neq \tau} \left( \sum_{\sigma \in T, \tau, \tau ' \in \sigma} w (\sigma) \right) =
\sum_{\sigma \in T, \tau \in \sigma} w (\sigma) \left( \sum_{\tau' \in E, \tau ' \neq \tau, v \in \tau ' \in \sigma} 1 \right).
\end{dmath*}
By our assumptions on $X$,
$$1 \leq \sum_{\tau' \in E, \tau ' \neq \tau, v \in \tau ' \in \sigma} 1 \leq s-1,$$
and thus $w (\tau) \leq m_v (\tau) \leq (s-1)w (\tau)$.

Also,
\begin{dmath*}
m_v (E_v) = \sum_{\tau \in E_v} m_v (\tau) =
{w (\lbrace \tau \in E : v \in \tau \rbrace)},
\end{dmath*}
and thus by Proposition \ref{w(v) as sum of w(tau)'s prop},
$$2 w(v) \leq m_v (E_v) \leq s w(v).$$
\end{proof}

For a set $A \subseteq E$, we define the \textit{localization of $A$ at $X_v$} to be the set $A_v = A \cap E_v$.
\begin{proposition}
\label{A_v and w(A) prop}
Let $A \subseteq E$ and $U \subseteq V$. For $0 \leq i \leq k$, define $A_U^i \subseteq A$ as
$$A_U^i = \lbrace \tau \in A : \vert U \cap \tau \vert = i \rbrace.$$
Then
$$\sum_{i=1}^k i w (A_U^i) \leq \sum_{v \in U} m_v (A_v) \leq (s-1) \sum_{i=1}^k i w (A_U^i).$$
In particular, for every set $A \subseteq E$ and every set $U \subseteq V$,
$$\sum_{v \in U} m_v (A_v) \leq (s-1) k w (A).$$
\end{proposition}

\begin{proof}[Proof of Proposition \ref{A_v and w(A) prop}]
Let $A  \subseteq E$ and $U$ as above. By Proposition \ref{m_v (tau), m_v (X_v) prop}, for every $\tau \in A$ and every $v \in \tau$, we have that
$w (\tau) \leq m_v (\tau) \leq (s-1) w (\tau)$. Thus
\begin{dmath*}
\sum_{v \in U} m_v (A_v) =
\sum_{v \in U} \left( \sum_{ \tau \in A, v \in \tau} m_v (\tau) \right) =
\sum_{v \in U} \left( \sum_{\tau \in A, v \in \tau} w (\tau) \right) \leq^{m_v (\tau) \leq (s-1) w (\tau)}
(s-1)\sum_{\tau \in A} w(\tau) \sum_{v \in U, v \in \tau} 1 =
(s-1) \sum_{i=1}^k \sum_{\tau \in A_U^i} w(\tau) \sum_{v \in U, v \in \tau} 1 =
(s-1)\sum_{i=1}^k \sum_{\tau \in A_U^i} i w(\tau) =
(s-1) \sum_{i=1}^k i w (A_U^i).
\end{dmath*}
and the proof that
$$\sum_{i=1}^k i w (A_U^i) \leq \sum_{v \in U} m_v (A_v)$$
is similar (only using $w (\tau) \leq m_v (\tau)$).

Note that
$$w(A) = \sum_{i=0}^k w (A_U^i),$$
and thus
\begin{dmath*}
\sum_{v \in U} m_v (A_v) \leq
(s-1) \sum_{i=1}^k i w (A_U^i) \leq
(s-1) \sum_{i=0}^k k w (A_U^i) = (s-1)k w(A),
\end{dmath*}
and the second assertion follows.
\end{proof}

\section{Codes modelled over two layer systems and local testability}
\label{Codes modelled over two layer systems and local testability sec}
As stated in the introduction, we show that a code modelled over a system $X$ that has $\delta$-unique neighbor expansion is amplified locally testable.

We start by giving the formal definition of a locally small set $A \subseteq E$:

\begin{definition}
Let $X$ be a $(s,k,K)$-two layer system, $A \subseteq E$ be a set and $0 \leq \mu \leq 1 $ be a constant. We will say that a vertex $v \in V$ is \textit{$\mu$-small} with respect to $A$ if
$$\frac{m_v (A_v)}{m_v (E_v)} < \mu.$$
Otherwise, we will say that $v$ is \textit{$\mu$-large}. with respect to $A$ When $A$ is obvious from the context, we will just use $\mu$-small and $\mu$-large.  We will denote $V_{\mu \text{-small}} \subseteq V$ to be the subset of $\mu$-small vertices and $V_{\mu \text{-large}} \subseteq V$ to be the subset of $\mu$-large vertices.
\end{definition}

\begin{definition}
\label{locally small def}
Let $X$ be a $(s,k,K)$-two layer system. For a set $A \subseteq E$ and constants $0 \leq \delta \leq 1, 0 \leq \alpha <1$, we say that $A$ is \textit{$(\delta,\alpha)$-locally small} if
$$\sum_{v \in V_{\delta \text{-large}}} m_v (A_v) \leq \alpha w (A).$$
We also say that $A$ is \textit{$\delta$-locally small} if it is \textit{$(\delta,0)$-locally small}, i.e., if $V_{\delta \text{-small}} = V$.
\end{definition}

\begin{definition}
\label{uni-neigh-exp def}
Let $X$ be a $(s,k,K)$-two layer system and let $0 < \delta < 1, 0 \leq \alpha <1, 0 < \varepsilon_0 <1$ be some constants. We say that $X$ has a $((\delta ,\alpha), \varepsilon_0)$-unique neighbor expansion property if for every non-empty set $A \subseteq E$, if $\frac{w(A)}{w(E)} < \varepsilon_0$ and $A$ is $(\delta,\alpha)$-locally small, then there is $\sigma \in T$ such that $\vert A \cap \sigma \vert =1$. We also say that $X$ has a  $(\delta, \varepsilon_0)$-unique neighbor expansion it is has $((\delta ,0), \varepsilon_0)$-unique neighbor expansion.
\end{definition}

Let $p$ be a prime power and $X = (V,E,T)$ be a $(s,k,K)$-two layer system. Also, let $C \subseteq \mathbb{F}_p^V$ be a linear code modelled over $X$ (see Definition \ref{code modelled over t-l-s def} above). Define the (normalized) norm on $\mathbb{F}_p^V$ induced by the weight function $w$ as
$$\Vert \underline{c} \Vert = \frac{1}{w(V)} \sum_{v, c (v) \neq 0} w (v).$$
We also define $w: \mathcal{E} \rightarrow \mathbb{R}$ induced by the weight function on $E$ as $w (\underline{e}) = w (\supp (\underline{e}))$ and for every $\mathcal{A} \subseteq \mathcal{E}$, we define $w (\mathcal{A}) = \sum_{\underline{e} \in \mathcal{A}} w (\underline{e})$.
Last, we define $\rej : \mathbb{F}_p^V \rightarrow [0,1]$ to be the function assigning each $\underline{c} \in \mathbb{F}_p^V$ the (weighted) fraction of the equation of $\mathcal{E}$ it violates: Explicitly, for every $\underline{c} \in \mathbb{F}_p^V$, we define
$$\mathcal{A} (\underline{c}) = \lbrace \underline{e} \in \mathcal{E} : \underline{e} \cdot \underline{c} \neq 0 \rbrace,$$
and we define $\rej (\underline{c}) = \frac{w( \mathcal{A} (\underline{c}))}{w (\mathcal{E})}$.

\begin{lemma}
\label{uni-neigh imply in C lemma}
Let $X = (V,E,T)$ be a $(s,k,K)$-two layer system and  $C \subseteq \mathbb{F}_p^V$ be a linear code modelled over $X$ with $(V, \mathcal{E}, \mathcal{T})$ as in the Definition \ref{code modelled over t-l-s def}. Assume that $X$ has $((\delta, \alpha), \varepsilon_0)$-unique neighbor expansion. For $\underline{c} \in \mathbb{F}_p^V$, define $A (\underline{c}) = \lbrace \supp (\underline{e}) : \underline{e} \in \mathcal{A} (\underline{c}) \rbrace$.
Given $\underline{c} \in \mathbb{F}_p^V$, if there is $0 <\varepsilon < \varepsilon_0$ such that $\frac{w (A (\underline{c}))}{w (E)} \leq \varepsilon$ and $A (\underline{c})$ is $(\delta, \alpha)$-locally small, then $\underline{c} \in C$.
\end{lemma}

\begin{proof}
Fix $0 <\varepsilon < \varepsilon_0$. Let $\underline{c} \in \mathbb{F}_p^V$ such that $\frac{w (A (\underline{c}))}{w (E)} \leq \varepsilon$ and $A (\underline{c})$ is $(\delta, \alpha)$-locally small.

We note that proving that $\underline{c} \in C$ is equivalent to proving that $\mathcal{A} (\underline{c})$ is an empty set. Assume towards contradiction that $\mathcal{A} (\underline{c}) \neq \emptyset$. Thus $A (\underline{c})$ and by the $((\delta, \alpha), \varepsilon_0)$-unique neighbor expansion there is some $\sigma \in T$ such that $\vert A (\underline{c}) \cap \sigma \vert = 1$.

Note that by Definition \ref{code modelled over t-l-s def} there is a bijection $\Phi : \mathcal{E} \rightarrow E$.
Thus, since $\tau$ is the support of some linear dependency, that there is a linear dependency $\ld \in \mathcal{T}$ such that $\vert \supp (\ld) \cap \mathcal{A} (\underline{c}) \vert =1$. Denote $\underline{e}_0 = \supp (\ld) \cap \mathcal{A} (\underline{c})$ and note that for every $\underline{e} \in \supp (\ld) \setminus \lbrace \underline{e}_0 \rbrace$ it holds that
$\underline{e} \cdot \underline{c} = 0$. Thus,
$$\sum_{\underline{e} \in \supp (\ld)} \ld (\underline{e}) (\underline{e} \cdot \underline{c}) = \ld (\underline{e}_0) (\underline{e}_0 \cdot \underline{c}) \neq 0,$$
but this is a contradiction to the fact that $\ld$ is a linear dependency.
\end{proof}

Below, we will show that for $s=2$ the assumption of unique neighbor expansion implies a variation of local testability in the sense that $\rej (\underline{c})$ can be bounded as a function of the distance of $\underline{c}$ to $C$:

\begin{theorem}
\label{unique neighbor exp imply LTC thm}
Let $X = (V,E,T)$ be a $(s,k,K)$-two layer system and $C \subseteq \mathbb{F}_p^V$ be a linear code modelled over $X$. Assume that $X$ has $(\delta, \varepsilon_0)$-unique neighbor expansion with $\delta > \frac{p-1}{p}$ and $\varepsilon_0 >0$.  Then for every $0 < \varepsilon < \varepsilon_0$ and every $\underline{c} \in  \mathbb{F}_p^V$, if $\rej (\underline{c}) \leq \varepsilon$, then
$$\min_{\underline{c} ' \in C} \Vert \underline{c} - \underline{c} ' \Vert \leq \frac{s}{k} \frac{\varepsilon}{2(\delta - \frac{p-1}{p})} .$$

\end{theorem}

Before proving Theorem \ref{unique neighbor exp imply LTC thm}, we will need the following lemma:
\begin{lemma}
\label{flip bit lemma}
Let $(V,\mathcal{E},\mathcal{T})$ be as in Definition \ref{code modelled over t-l-s def}. For any $v_{0} \in V$ and any $\underline{c} \in \mathbb{F}_p^V$ there is $\underline{c} '  \in \mathbb{F}_p^V$ such that for every $v \neq v_0$, $c (v) = c ' (v)$ and
$$\frac{m_{v_0} (\lbrace \supp (\underline{e} ) : \supp (\underline{e} ) \in E_{v_{0}}, \underline{e} \cdot \underline{c} ' \neq 0  \rbrace)}{m_{v_0} (E_{v_{0}})} \leq \frac{p-1}{p}.$$
\end{lemma}

\begin{proof}
Fix $v_{0} \in V$ and  $\underline{c} \in \mathbb{F}_p^V$. For $t \in \mathbb{F}_p$, define $\underline{c}^{+t} \in \mathbb{F}_p^V$ as
$$\underline{c}^{+t} (v) = \begin{cases}
c (v_0) + t & v = V_0 \\
c (v) & v \neq V_0
\end{cases}.$$
Note that for every $\underline{e} \in \mathcal{E}$ with $ \supp (\underline{e}) \in E_{v_0}$ there is a unique $t$ such that $\underline{e} \cdot \underline{c}^{+t} = 0$. It follows that
$$\sum_{t \in \mathbb{F}_p} m_{v_0} (\lbrace \supp (\underline{e}) : v_{0} \in \supp (\underline{e}), \underline{e} \cdot \underline{c}^{+t} = 0 \rbrace) = m_{v_0} (E_{v_0}),$$
thus there is $t \in \mathbb{F}_p$ with
$$\frac{m_{v_0}(\lbrace \supp (\underline{e}) : v_{0} \in \supp (\underline{e}), \underline{e} \cdot \underline{c}^{+t} = 0 \rbrace)}{m_{v_0}(E_{v_0})} \geq \frac{1}{p},$$
or equivalently,
$$\frac{m_{v_0}(\lbrace \supp (\underline{e}) : v_{0} \in \supp (\underline{e}), \underline{e} \cdot \underline{c}^{+t} \neq 0 \rbrace)}{m_{v_0}(E_{v_0})} \leq \frac{p-1}{p}.$$
\end{proof}

Next, we will prove Theorem \ref{unique neighbor exp imply LTC thm}:
\begin{proof}



Fix $0 <\varepsilon < \varepsilon_0$ and $\underline{c} \in \mathbb{F}_p^V$ with $\rej (\underline{c}) \leq \varepsilon$.

We define the following iterative procedure: Denote $\underline{c}^0 = \underline{c}$ and for every $j \in \mathbb{N} \cup \lbrace 0 \rbrace$, define $A^j = A (\underline{c}^j)$.
For $j \in \mathbb{N} \cup \lbrace 0 \rbrace$, if $A^j$ is $\delta$-locally small (i.e., if all the vertices are $\delta$-small) then stop and denote $\underline{c} ' = \underline{c}^j$. Otherwise, there is some $v^{(j)} \in V$ such that
$$\frac{m_{v^{(j)}} (\lbrace \supp (\underline{e}) : \underline{e} \cdot \underline{c}^j \neq 0, \supp (\underline{e}) \in E_{v^{(j)}} \rbrace)}{m_{v^{(j)}} (E_{v^{(j)}})} \geq \delta.$$
By Lemma \ref{flip bit lemma}, we change the value of $\underline{c}^j$ only at $v^{(j)}$ to produce a new word $\underline{c}^{j+1}$ such that
$$\frac{m_{v^{(j)}} (\lbrace \supp (\underline{e}) : \underline{e} \cdot \underline{c}^j \neq 0, \supp (\underline{e}) \in E_{v^{(j)}} \rbrace)}{m_{v^{(j)}} (E_{v^{(j)}})} \leq \frac{p-1}{p}.$$
Note that by the definition of the procedure, $A^j$ and $A^{j+1}$ differ only at $E_{v^{(j)}}$ and
\begin{dmath*}
w(\lbrace \supp (\underline{e}) : \underline{e} \cdot \underline{c}^j \neq 0, \supp (\underline{e}) \in E_{v^{(j)}} \rbrace) - w(\lbrace \supp (\underline{e}) : \underline{e} \cdot \underline{c}^{j+1} \neq 0, \supp (\underline{e}) \in E_{v^{(j)}} \rbrace) \geq (\delta - \frac{p-1}{p}) m_{v^{(j)}}  (E_{v^{(j)}}).
\end{dmath*}
Thus,
$$w(A^{j+1}) + (\delta - \frac{p-1}{p})  m_{v^{(j)}}  (E_{v^{(j)}}) \leq  w(A^{j}),$$
and by Proposition \ref{m_v (tau), m_v (X_v) prop},
$$w(A^{j+1}) + 2(\delta - \frac{p-1}{p})  w (v^{(j)}) \leq  w(A^{j}).$$
It follows that for every $j$, $w(A^{j}) > w(A^{j+1})$ and thus the procedure terminates after finitely many steps which we will denote by $N$. Observe the following: First, $w(A^{0}) \geq  w(A^{N})$ and thus for $\underline{c}' = \underline{c}^N$,
$$\frac{w(\mathcal{A} (\underline{c} '))}{w(E)} \leq \varepsilon.$$
Second, $\underline{c}'$ is $\delta$-locally small. It follows from Lemma \ref{uni-neigh imply in C lemma} that $\underline{c}' \in C$.

We will finish the proof by showing that $\Vert \underline{c} - \underline{c} ' \Vert \leq \frac{s}{k} \frac{\varepsilon}{2(\delta - \frac{p-1}{p})}$:
\begin{dmath*}
{\Vert \underline{c} - \underline{c} ' \Vert} =
\frac{1}{w(V)} \sum_{v \in V, c (v) \neq c ' (v)} w(v) \leq
\frac{1}{w(V)} \sum_{j=0}^{N-1} w(v^{(j)}) \leq \\
\frac{1}{2(\delta - \frac{p-1}{p})} \frac{1}{w(V)} \left( w (A^N) + \sum_{j=0}^{N-1} 2(\delta - \frac{p-1}{p}) w(v^{(j)}) \right)  \leq \\
\frac{1}{2(\delta - \frac{p-1}{p})} \frac{1}{w(V)} \left( w (A^N) + \sum_{j=0}^{N-1} w (A^j) - w (A^{j+1}) \right) =
\frac{1}{2(\delta - \frac{p-1}{p})} \frac{1}{w(V)} w(A^0) = \\
\frac{1}{2(\delta - \frac{p-1}{p})}  \frac{w(E)}{w(V)} \frac{w(A)}{w(E)}  \leq^{\text{Corollary } \ref{w(V) and w(E) coro}} \\
\frac{1}{2(\delta - \frac{p-1}{p})} \frac{s}{k} \varepsilon = \frac{s}{k} \frac{\varepsilon}{2(\delta - \frac{p-1}{p})},
\end{dmath*}
as needed.
\end{proof}

Theorem \ref{unique neighbor exp imply LTC thm} can be rephrased as amplified local testability:
\begin{corollary}
\label{unique neighbor exp imply amp. loc. test. coro}
For every $p$ prime, $\delta > \frac{p-1}{p}$ $t' \in \mathbb{N}$ and $\mu$ constants, let $\mathcal{C} (\delta, p, t' ,\mu) $ be the family of $p$-ary codes (i.e.., codes of the form $C \subseteq \mathbb{F}_p^{V (C)}$ ) modelled over two layer systems such that
\begin{dmath*}
\mathcal{C} (\delta, p, t' ,\mu) =
{\left\lbrace C : \exists \varepsilon_0 (C) >0 \text{ such that } C \text{ has the } (\delta, \varepsilon_0 (C)) \text{-unique neighbor property and } \varepsilon_0 \geq \frac{\mu}{k(C)^{t'} } \right\rbrace.}
\end{dmath*}
Then the family $\mathcal{C} (\delta, p, t' ,\nu)$ is amplified locally testable with $r_{\mathcal{C} (\delta, p, t' ,\nu)} = \frac{2  \mu (\delta - \frac{p-1}{p})}{s}$ and $t_{\mathcal{C} (\delta, p, t' ,\nu)} = t' +1$.  In other words, for every $C \in \mathcal{C} (\delta, p, t' ,\nu)$ and every $\underline{c} \in \mathbb{F}_p^{V(C)}$ it holds that
$$\rej (\underline{c}) \geq k(C) \frac{2  \mu (\delta - \frac{p-1}{p})}{s} \min \left\lbrace \min_{\underline{c} ' \in C} \Vert \underline{c} - \underline{c} ' \Vert, \frac{1}{k (C)^{t'+1}} \right\rbrace.$$
\end{corollary}

\begin{proof}
Let $C \subseteq \mathbb{F}_p^{V(C)}$ such that $C \in \mathcal{C} (\delta, p, t' ,\nu)$ and let $\underline{c} \in \mathbb{F}_p^{V(C)}$. If $\rej (\underline{c}) = \varepsilon < \varepsilon_0$, then by Theorem \ref{unique neighbor exp imply LTC thm},
\begin{dmath*}
\min_{\underline{c} ' \in C} \Vert \underline{c} - \underline{c} ' \Vert \leq \frac{s}{k} \frac{\varepsilon}{2(\delta - \frac{p-1}{p})} =
\frac{s}{k} \frac{\rej (\underline{c})}{2(\delta - \frac{p-1}{p})},
\end{dmath*}
i.e.,
$$\rej (\underline{c}) \geq k \frac{2  (\delta - \frac{p-1}{p})}{s} min_{\underline{c} ' \in C} \Vert \underline{c} - \underline{c} ' \Vert \geq
 k \frac{2  \mu (\delta - \frac{p-1}{p})}{s} min_{\underline{c} ' \in C} \Vert \underline{c} - \underline{c} ' \Vert .$$
Assume that $\rej (\underline{c}) \geq \varepsilon_0 (C)$, then by our assumption on $\varepsilon_0 (C)$,
$$ \rej (\underline{c})\geq \frac{\mu}{k(C)^{t'}} \geq k (C) \frac{2  \mu (\delta - \frac{p-1}{p})}{s} \frac{1}{k (C)^{t' +1}}.$$
\end{proof}

\section{Main Expansion Theorem}

\label{Main Expansion Theorem sec}

Below we prove the main expansion theorem of this paper, namely that a $\lambda$-expanding $\HDE$ is has unique neighbor expansion (given that $\lambda$ is sufficiently small). We start by recalling some definitions and results concerning expansion of weighted graphs.

\subsection{Expansion of weighted graphs - basic definitions and results}
The aim of this section is to review the basic definitions regarding weighted graph and to state the Cheeger inequality and Alon-Chung Lemma in this setting. This section does not contain any new results and we provide the proofs merely for the sake of completeness.

\label{Weighted graphs subsec}

Let $G = (V,E)$ be a weighted graph with a weight function $m$. For $\emptyset \neq U_1, U_2 \subseteq V$, we denote
$$m(U_1,U_2 ) = \sum_{(u_1,u_2) \in U_1 \times U_2,  \lbrace u_1, u_2 \rbrace \in E} m (\lbrace u_1, u_2 \rbrace ).$$


We will work with the following variation on the Cheeger constant:
\begin{definition}
\label{Cheeger constant def}
Let $G = (V,E)$ be a finite graph. The (generalized) Cheeger constant of $G$ is
$$h_G = \min_{\emptyset \neq U \subsetneqq V} \frac{m(U,V \setminus U) m(V)}{m (U) m (V \setminus U)}.$$
\end{definition}
Note that $m(U,V \setminus U) \leq \min \lbrace m(U), m (V \setminus U) \rbrace$ and therefore $h_G \leq 2$.

We will use the following (unorthodox) definition of a $\lambda$-expander:
\begin{definition}[$\lambda$-expander]
A graph $G$ will be called a $\lambda$-expander if $1-h_G \leq \lambda$.
\end{definition}
The Cheeger inequality shows that this definition follows from $\lambda$-spectral expansion - see Theorem \ref{The Cheeger ineq thm} below.

\begin{proposition}
\label{m(U,U) prop}
Let $G$ be a $\lambda$-expander with $\lambda \geq 0$, then for every $\emptyset \neq U \subsetneqq V$,
$$m(U) \left(\lambda + \frac{m(U)}{m(V)} \right) \geq m (U,U).$$
\end{proposition}

\begin{proof}[Proof of Proposition \ref{m(U,U) prop}]
Let $\emptyset \neq U \subsetneqq V$, then
$$m(U) = m(U, V \setminus U) + m(U,U).$$
Thus, by the definition of $h_G$,
\begin{dmath*}
h_G \frac{m(U) m(V \setminus U)}{m(V)} \leq m(U, V \setminus U) = m(U) - m(U,U).
\end{dmath*}
It follows that
$$m(U) \left(1 - h_G +  h_G \frac{m(U)}{m(V)} \right) = m(U)  \left(1 - h_G \frac{m(V \setminus U)}{m(V)} \right) \geq m (U,U),$$
and substituting $h_G$ by $1- \lambda$ and using the fact that $\lambda \geq 0$ finished the proof.
\end{proof}

Below, we define the (weighted) random walk on a weighted graph and state the Cheeger inequality in this setting.

Define $\ell^2 (V)$ to be the space of functions $\lbrace \phi : V \rightarrow \mathbb{R} \rbrace$ with the inner-product
$$\langle \phi, \psi \rangle = \sum_{v \in V} m(v) \phi (v) \psi (v),$$
and the induced norm
$$\Vert \phi \Vert^2 = \sum_{v \in V} m(v) \phi (v)^2.$$
The random walk on the $(V,E)$ is the operator $M : \ell^2 (V) \rightarrow \ell^2 (V)$ defined by
$$M \phi (v) = \frac{1}{m(v)} \sum_{u, \lbrace u,v \rbrace \in E} m (\lbrace u,v \rbrace) \phi (u).$$
We note that when $m :E \rightarrow \mathbb{R}^+$ is the constant function $1$, then $M$ is the simple random walk on the graph.

The standard facts regarding this operator are:
\begin{itemize}
\item The operator $M$ is self-adjoint and of operator norm $1$ with respect to the inner product on $\ell^2 (V)$ defined above.
\item For every $\phi \in \ell^2 (V)$,
$$\langle (I-M) \phi, \phi \rangle = \sum_{\lbrace u,v \rbrace \in E} m(\lbrace u,v \rbrace) (\phi (u) - \phi (v) )^2.$$
\item For the constant function $\mathbbm{1} \in \ell^2 (V)$ it holds that $M \mathbbm{1} = \mathbbm{1}$, i.e., this is an eigenfunction with the eigenvalue $1$.
\item The graph $(V,E)$ is connected if and only if $1$ is an eigenvalue of $M$ with multiplicity $1$.
\end{itemize}

\begin{theorem}[The Cheeger inequality, Alon-Chung Lemma]
\label{The Cheeger ineq thm}
Let $(V,E)$ be a connected graph with a weight function $m$. Denote by $\lambda$ the second largest eigenvalue of $M$. Let $\emptyset \neq U \subsetneqq V$ be a set of vertices. Then the following inequalities hold:
\begin{enumerate}
\item The Cheeger inequality:
$$h_G \geq 1-\lambda.$$
In other words, $G$ is $\lambda$-Cheeger-expanding.
\item Alon-Chung Lemma:
$$ m(U)  \left(\lambda + (1-\lambda) \frac{m(U)}{m(V)} \right) \geq m (U,U)  .$$
\end{enumerate}
\end{theorem}

\begin{proof}[Proof of Theorem \ref{The Cheeger ineq thm}]
We start by proving The Cheeger inequality. For a set $U$ as above, define $\phi \in \ell^2 (V)$ by
$$\phi (v) = \begin{cases}
m (V \setminus U) & v \in U \\
- m(U) & v \in V \setminus U
\end{cases}.$$
We observe that $\phi \perp \mathbbm{1}$, indeed:
\begin{dmath*}
\langle \phi, \mathbbm{1} \rangle = \sum_{v \in U} m(v) m (V \setminus U) + \sum_{v \in V \setminus U} m(v) (- m (U)) = m(U) m(V \setminus U) - m(V \setminus U) m(U)  = 0.
\end{dmath*}
Thus,
$$\langle (I-M) \phi, \phi \rangle \geq (1- \lambda) \Vert \phi \Vert^2.$$
A direct computation shows that
$$\Vert \phi \Vert^2 = m(U) m (V \setminus U) (m(U) + m (V \setminus U)) = m(U) m (V \setminus U) m (V).$$
By the fact stated above regarding $\langle (I-M) \phi, \phi \rangle$, we get that
\begin{dmath*}
(1- \lambda) m(U) m (V \setminus U) (m(V)) \leq \langle (I-M) \phi, \phi \rangle =
\sum_{\lbrace u,v \rbrace \in E} m(\lbrace u,v \rbrace) (\phi (u) - \phi (v) )^2 =
\sum_{\lbrace u,v \rbrace \in E, u \in U, v \in V \setminus U} m(\lbrace u,v \rbrace) (m(U) + m(V \setminus U) )^2 = m(U, V \setminus U)  m (V)^2.
\end{dmath*}
Dividing this inequality by $m(V)^2$ yields the Cheeger inequality stated in the Theorem.

The Alon-Chung Lemma is a consequence of the Cheeger inequality and Proposition \ref{m(U,U) prop} above.
\end{proof}

We note that in the case where the graph is ``almost complete'', there is an easy bound on the Cheeger constant:
\begin{lemma}
\label{almost complete graph lemma}
Let $G=(V,E)$ be a weighted graph. If there is a constant $0 \leq \beta < 1$ such that for every $v \in V$, such that $m(v) \geq (\max_{e' \in E} m(e'))(1-\beta) \vert V \vert$, then $h_G \geq 1- 2 \beta$, i.e., $G$ is $2 \beta$-expander.
\end{lemma}

\begin{proof}[Proof of Lemma \ref{almost complete graph lemma}]
We start by defining a normalized weight function $m ' (e) = \frac{m(e)}{\max_{e' \in E} m(e')}$. With this normalization, we get that for every $e$, $m ' (e) \leq 1$ and for every $v \in V$, $m '(v) \geq (1-\beta) \vert V \vert$. We note that this normalization does not change the Cheeger constant of the graph.

Let $U \subseteq V$ be some non-empty subset such that $\vert U \vert \leq \frac{1}{2}$. By the assumption $m ' (v) \geq (1-\beta) \vert V \vert$ it follows that
$m ' (V) \geq (1-\beta) \vert V \vert^2$. Also,
\begin{dmath*}
m ' (U, V \setminus U) = \sum_{u \in U} m(\lbrace u \rbrace , V \setminus U) \geq
\sum_{u \in U} (m ' (u) - m ' (\lbrace u \rbrace , U)) \geq
\sum_{u \in U} (m ' (u) - \vert U \vert) \geq
\vert U \vert \left( (1-\beta) \vert V \vert - \vert U \vert \right).
\end{dmath*}
Note that for every $v \in V$, $m ' (v) \leq \vert V \vert$ and thus
$$m ' (U) \leq \vert V \vert \vert U \vert, m ' (V \setminus U) \leq \vert V \vert (\vert V \vert - \vert U \vert).$$
From the above inequalities, it follows that
\begin{dmath*}
\frac{m  (U,V \setminus U) m  (V)}{m  (U) m  (V \setminus U)} = \frac{m ' (U,V \setminus U) m ' (V)}{m ' (U) m ' (V \setminus U)} \geq
\frac{\vert U \vert ((1-\beta)\vert V \vert - \vert U \vert) (1-\beta) \vert V \vert^2}{\vert V \vert \vert U \vert \vert \vert V \vert (\vert V \vert - \vert U \vert)} =
1 - \frac{\beta \vert V \vert}{\vert V \vert - \vert U \vert} \geq 1- \frac{\beta \vert V \vert}{\vert V \vert - \frac{1}{2}\vert V \vert} = 1 - 2 \beta.
\end{dmath*}
The above inequality holds for every non-empty $U \subseteq V$ with $\vert U \vert \leq \frac{1}{2} \vert V \vert$ and thus $h_G \geq 1 - 2 \beta$.
\end{proof}

Below, in applications, we will also need a result for the Cheeger constant of (weak) graph covers.
\begin{definition}[Weak cover]
\label{weak covering def}
Given two finite weighted graphs $G' = (V',E',m'), G = (V,E,m)$ we say that $G'$ is a weak cover of $G$ if there is a surjective function $p :V' \rightarrow V$ such that
\begin{itemize}
\item For any $\lbrace v', u' \rbrace \in E'$ it follows that $\lbrace p(v'), p(u') \rbrace \in E$ and the induced map $p: E' \rightarrow E$ is surjective.
\item For every $e \in E$, $m(e) = \sum_{e' \in p^{-1} (e)} m(e')$.
\end{itemize}

\end{definition}

\begin{proposition}
\label{weak cover cheeger bound prop}
Let $G' = (V',E',m'), G = (V,E,m)$ be two finite weighted graphs such that $G'$ is a weak cover of $G$. Then $h_G \geq h_{G'}$.
\end{proposition}

\begin{proof}
By the definition of a weak cover, we have for every two sets $U_1, U_2 \subseteq V$ that
\begin{dmath*}
m(U_1, U_2) =
\sum_{(u_1,u_2) \in U_1 \times U_2,  \lbrace u_1, u_2 \rbrace \in E} m (\lbrace u_1, u_2 \rbrace ) =
\sum_{(u_1,u_2) \in U_1 \times U_2,  \lbrace u_1, u_2 \rbrace \in E} \sum_{u_1 '  \in p^{-1} (u_1), u_2 '  \in p^{-1} (u_2), \lbrace u_1 ' , u_2' \rbrace \in E'} m (\lbrace u_1 ', u_2 ' \rbrace ) =
\sum_{(u_1 ',u_2 ') \in p^{-1} (U_1) \times p^{-1} (U_2),  \lbrace u_1 ', u_2 ' \rbrace \in E '} m (\lbrace u_1 ', u_2 ' \rbrace ) =
m' (p^{-1} (U_1), p^{-1} (U_2)).
\end{dmath*}
Thus it follows that for every non-empty set $U \subseteq V$,
$$m(U, V \setminus U) = m' (p^{-1} (U), p^{-1} (V \setminus U)) = m' (p^{-1} (U), V' \setminus p^{-1} (U))$$
and
$$m (U) = m' (p^{-1} (U)), m (V \setminus U) = m' (V' \setminus p^{-1} (U)),  m (V) = m' (V ').$$
Thus for every non-empty set $U \subseteq V$, if we denote $U' = p^{-1} (U)$ it holds that
$$\frac{m(U,V \setminus U) m(V)}{m (U) m (V \setminus U)} = \frac{m' (U', V' \setminus U') m ' (V ')}{m ' (U') m ' (V' \setminus U')} \geq h_{G'}$$
and it follows that $h_G \geq h_{G'}$.
\end{proof}

\subsection{Main Expansion Theorem}

After the above definition of $\lambda$-expanders, we can give the formal definition of an $\HDE$:

\begin{definition}
\label{HDE formal def}
Let $X$ be a $(s,k,K)$-two layer system. For a constant $\lambda <1$, we say that $X$ is a $\lambda$-expanding $\HDE$, if the ground graph and the links of all the vertices are all $\lambda$-expanders and the non-intersecting graph is either a totally disconnected (i.e., contains no edges) or a $\lambda$-expander.
\end{definition}

Our main theorem is regarding HDE-systems is for every $k,K \in \mathbb{N}$ and every $0<\delta <1$ there is a $\lambda (k,K,\delta) >0$ such that every $(s,k,K)$-two layer system that is a $\lambda$-expanding $\HDE$ has $\delta$-unique neighbor expansion. We state (and prove) this theorem formally below, but before doing so, we will need some additional lemmas.

\begin{lemma}
\label{ground graph lemma}
Let $\lambda_\ground \geq 0$ be a constant and $X$ a $(s,k,K)$-two layer system such that the ground graph of $X$ is a $\lambda_\ground$-expander. Also, let $A \subseteq E$ be a non-empty set, $0<\mu < 1$ be a constant and $U \subseteq V_{\mu \text{-large}}$. For $0 \leq i \leq k$, we define
$$ A_{U}^i = \lbrace \tau \in A : \vert \tau \cap U \vert = i \rbrace.$$
If $\frac{w(A)}{w(E)} \leq \frac{4 \mu^2}{s^3 (s-1)^2}$, then
\begin{dmath*}
\frac{s(s-1)(k-1)}{2 \mu} \lambda_\ground +  \frac{s^3 (s-1)^2}{4 \mu^2} \frac{w(A)}{w(E)} \geq
 \left(1-\frac{s(s-1)(k-1)}{2 \mu} \lambda_\ground  \right) \sum_{i =2}^k  (i-1) \frac{w(A_{U}^i)}{w(A)}.
\end{dmath*}
\end{lemma}

\begin{proof}
By Proposition \ref{m(U,U) prop},
\begin{equation}
\label{eq1}
m_\ground (U) \left(\lambda_\ground +  \frac{m_\ground (U)}{m_\ground (V)} \right) \geq m_\ground (U,U).
\end{equation}

By Proposition \ref{m-ground (U,U) prop},
\begin{dmath}
\label{eq2}
m_\ground (U,U) = \sum_{i =2}^k i (i-1) w(E^i_{U}) \geq
\sum_{i =2}^k i (i-1) w(A_{U}^i).
\end{dmath}

By Proposition \ref{m_gr (v) prop} and Corollary \ref{w(V) and w(E) coro},
\begin{equation}
\label{eq3}
m_\ground (V) \geq 2 (k-1) w (V) \geq \frac{2 k (k-1)}{s} w (E).
\end{equation}

By the fact that of $U \subseteq V_{\mu \text{-large}}$, for every $v \in U$ we have that $\frac{m_v (A_v)}{m_v (E_v)} \geq \mu$.
Using Proposition \ref{A_v and w(A) prop}:
\begin{dmath}
\label{eq4}
\sum_{i=1}^k i w (A_U^i) \geq
\frac{1}{s-1} \sum_{v \in  U} m_{v} (A_v) \geq
\frac{1}{s-1} \sum_{v \in  U} \mu m_{v} (E_v) \geq^{\text{Proposition }\ref{m_v (tau), m_v (X_v) prop}}
\frac{1}{s-1} \sum_{v \in  U} 2 \mu w (v) =
\frac{2}{s-1}  \mu w ( U) \geq^{\text{Proposition } \ref{m_gr (v) prop}}
\frac{2}{s-1}  \frac{\mu}{s(k-1)} m_\ground (U).
\end{dmath}
Using the fact that
\begin{equation}
\label{eq5}
w(A) = \sum_{i=0}^k w (A_{U}^i),
\end{equation}
and \eqref{eq4}, we have that
\begin{equation}
\label{eq6}
\frac{s(s-1)(k-1)}{2 \mu} \left( w(A) + \sum_{i=2}^k (i-1) w (A_{U}^i) \right) \geq   m_\ground (U).
\end{equation}
Using this inequality and \eqref{eq3},
\begin{dmath}
\label{eq7}
m_\ground (U) \left(\lambda_\ground +  \frac{m_\ground (U)}{m_\ground (V)} \right) \\ \leq
\frac{s(s-1)(k-1)}{2 \mu} \left( w(A) + \sum_{i=2}^k (i-1) w (A_{U}^i) \right) \left(\lambda_\ground + \frac{s(s-1)(k-1)}{2 \mu} \frac{s}{2 k(k-1) w(E)} \left( w(A) + \sum_{i=2}^k (i-1) w (A_{U}^i) \right) \right) =
 w(A)  \left(1 + \sum_{i=2}^k (i-1) \frac{w (A_{U}^i)}{w(A)} \right) \left(\frac{s(s-1)(k-1)}{2 \mu}  \lambda_\ground +  \frac{s^3 (s-1)^2 (k-1)}{8 \mu^2 k}  \frac{w(A)}{w(E)}  \left( 1 + \sum_{i=2}^k (i-1) \frac{w (A_{U}^i)}{w(A)} \right) \right) \leq
 w(A)  \left(1 + \sum_{i=2}^k (i-1) \frac{w (A_{U}^i)}{w(A)} \right) \left(\frac{s(s-1)(k-1)}{2 \mu} \lambda_\ground + \frac{s^3 (s-1)^2}{8 \mu^2} \frac{w(A)}{w(E)} \left( 1 + \sum_{i=2}^k (i-1) \frac{w (A_{U}^i)}{w(A)} \right) \right).
\end{dmath}
Combining this inequality with \eqref{eq1}, \eqref{eq2} yields
\begin{dmath*}
 w(A)  \left(1 + \sum_{i=2}^k (i-1) \frac{w (A_{U}^i)}{w(A)} \right) \left(\frac{s(s-1)(k-1)}{2 \mu} \lambda_\ground +   \frac{s^3 (s-1)^2}{8 \mu^2} \frac{w(A)}{w(E)} \left( 1 + \sum_{i=2}^k (i-1) \frac{w (A_{U}^i)}{w(A)} \right) \right) \geq
 \sum_{i =2}^k i (i-1) w(A_{U}^i).
\end{dmath*}
After dividing by $w(A)$ and re-arranging, we get
\begin{dmath}
\label{eq8}
\frac{s(s-1)(k-1)}{2 \mu} \lambda_\ground +  \frac{s^3 (s-1)^2}{8 \mu^2} \frac{w(A)}{w(E)} \left( 1 + \sum_{i=2}^k (i-1) \frac{w (A_{U}^i)}{w(A)} \right)^2 \geq
 \sum_{i =2}^k \left(i-\frac{s(s-1)(k-1)}{2 \mu} \lambda_\ground \right) (i-1) \frac{w(A_{U}^i)}{w(A)}.
\end{dmath}
We note that
\begin{dmath}
\label{eq9}
\left( 1 + \sum_{i=2}^k (i-1) \frac{w (A_{U}^i)}{w(A)} \right)^2 =
\left( 1 \cdot 1 + \sum_{i=2}^k \sqrt{\frac{w (A_{U}^i)}{w(A)}} \cdot (i-1) \sqrt{\frac{w (A_{U}^i)}{w(A)}} \right)^2 \leq^{\text{Cauchy-Schwarz}}
\left(1 + \sum_{i=2}^k \frac{w (A_{U}^i)}{w(A)} \right) \left(1 + \sum_{i=2}^k (i-1)^2 \frac{w (A_{U}^i)}{w(A)} \right) \leq
2 \left(1 + \sum_{i=2}^k (i-1)^2 \frac{w (A_{U}^i)}{w(A)} \right),
\end{dmath}
where the last inequality is due to \eqref{eq5}. Combining inequalities \eqref{eq8} and \eqref{eq9} yields
\begin{dmath*}
\frac{s(s-1)(k-1)}{2 \mu} \lambda_\ground +  \frac{s^3 (s-1)^2}{4 \mu^2} \frac{w(A)}{w(E)} \left( 1 + \sum_{i=2}^k (i-1)^2 \frac{w (A_{U}^i)}{w(A)} \right) \geq
 \sum_{i =2}^k \left(i-\frac{s(s-1)(k-1)}{2 \mu} \lambda_\ground \right) (i-1) \frac{w(A_{U}^i)}{w(A)}.
\end{dmath*}
After re-arranging, we get
\begin{dmath*}
\frac{s(s-1)(k-1)}{2 \mu} \lambda_\ground +  \frac{s^3 (s-1)^2}{4 \mu^2} \frac{w(A)}{w(E)} \geq
 \sum_{i =2}^k \left(i-\frac{s(s-1)(k-1)}{2 \mu} \lambda_\ground - \frac{s^2 (s-1)^2}{2 \mu^2} \frac{w(A)}{w(E)} (i-1) \right) (i-1) \frac{w(A_{U}^i)}{w(A)}.
\end{dmath*}
Recall we assumed that $\frac{s^3 (s-1)^2}{4 \mu^2} \frac{w(A)}{w(E)} \leq 1$ and therefore the above inequality yields
\begin{dmath*}
\frac{s(s-1)(k-1)}{2 \mu} \lambda_\ground +  \frac{s^3 (s-1)^2}{4 \mu^2} \frac{w(A)}{w(E)} \geq
 \left(1-\frac{s(s-1)(k-1)}{2 \mu} \lambda_\ground  \right) \sum_{i =2}^k  (i-1) \frac{w(A_{U}^i)}{w(A)},
\end{dmath*}
as needed.

\end{proof}

\begin{lemma}
\label{links lemma}
Let $\lambda_\loc \geq 0$ be a constant and $X$ be a $(s,k,K)$-two layer system such that all the links are $\lambda_\loc$-expanders. Also, let $A \subseteq E$ be a non-empty set, $0< \mu <1$ be a constant and  $U \subseteq V_{\mu \text{-small}}$ be a set. Then
$$(s-1) \left(\lambda_\loc + \mu  \right)  \sum_{i=1}^k i w (A_U^i) \geq \sum_{v \in U} m_v (A_v,A_v).$$
\end{lemma}

\begin{proof}
By Proposition \ref{m(U,U) prop}, for every $v \in U$,
$$m_v (A_v) \left(\lambda_\loc + \frac{m_v (A_v)}{m_v (E_v)} \right) \geq m_v (A_v,A_v).$$
By the assumption that $U \subseteq V_{\mu \text{-small}}$,
$$\frac{m_v (A_v)}{m_v (E_v)} \leq \mu.$$
Summing over all $v \in U$ yields
$$\sum_{v \in U} m_v (A_v) \left(\lambda_\loc + \mu   \right) \geq \sum_{v \in U} m_v (A_v,A_v).$$
By Proposition \ref{A_v and w(A) prop},
$$\sum_{v \in U} m_v (A_v) \leq (s-1) \sum_{i=1}^k i w (A_U^i),$$
and the needed inequality follows.
\end{proof}

\begin{lemma}
\label{non-int graph lemma}
Let $\lambda_\nint \geq 0$ be a constant and $X$ be a $(s,k,K)$-two layer system such that the non-intersecting graph of $X$ is $\lambda_\nint$-expander. For a non-empty set $A \subseteq E$, denote
$$D_{\nint}^{\geq 2} = \lbrace \sigma \in T : \exists \tau, \tau ' \in A \cap \sigma, \tau \leftrightarrow \tau ' \rbrace.$$
Then
$$\frac{1}{2 R_\nint}  \left(\lambda_\nint + \frac{1}{R_\nint} \frac{w (A)}{w (E)} \right) \geq \frac{w (D_{\nint}^{\geq 2})}{w (A)}.$$
\end{lemma}

\begin{proof}
Recall that  $Q_{\nint}^\max, Q_{\nint}^\min$ denote the constants such that for every $\sigma \in T$ and every $\tau \in \sigma$,
$$Q_{\nint}^\min \leq \vert \lbrace \tau ' \in \sigma : \tau ' \leftrightarrow \tau \rbrace \vert \leq Q_{\nint}^\max.$$

By Proposition \ref{m(U,U) prop},
$$m_\nint (A) \left(\lambda_\nint + \frac{m_\nint (A)}{m_\nint (E)} \right) \geq m_\nint (A,A).$$
By Proposition \ref{Q opp prop} and Corollary \ref{frac of m_opp coro}, we have that
\begin{dmath*}
m_\nint (A) \left(\lambda_\nint +  \frac{m_\nint (A)}{m_\nint (E)} \right) \leq
Q_{\nint}^\max w (A) \left(\lambda_\nint + \frac{1}{R_\nint} \frac{w (A)}{w (E)} \right).
\end{dmath*}
Thus we are left to prove that
$$m_\nint (A,A) \geq 2 Q_{\nint}^\min w( D_{\nint}^{\geq 2}).$$
Indeed,
\begin{dmath*}
m_\nint (A,A) =
\sum_{\tau \in A} \left( \sum_{\tau ' \in A, \tau \leftrightarrow \tau '} \left( \sum_{\sigma \in T, \tau, \tau ' \in \sigma} w(\sigma) \right) \right) =
\sum_{\sigma \in T} w(\sigma) \left( \sum_{\tau \in A, \tau \in \sigma} \left(  \sum_{\tau ' \in A, \tau \leftrightarrow \tau ', \tau ' \in \sigma} 1 \right) \right) =
 \sum_{\sigma \in D_{\nint}^{\geq 2}} w(\sigma) \left( \sum_{\tau \in A, \tau \in \sigma} \left(  \sum_{\tau ' \in A, \tau \leftrightarrow \tau ', \tau ' \in \sigma} 1 \right) \right) =
  \sum_{\sigma \in D_{\nint}^{\geq 2}} w(\sigma) \left( \sum_{\tau \in A, \tau \in \sigma} Q_{\nint}^\min \right) \geq
  2 Q_{\nint}^\min w( D_{\nint}^{\geq 2}).
\end{dmath*}
\end{proof}

After this, we can state and prove our main expansion theorem:
\begin{theorem}
\label{main exp thm - detailed}
Let $s,k,K \in \mathbb{N}$, $0 < \delta <\frac{1}{s-1}$, $0 \leq \alpha <1$ be constants.  Let $X$ be a $(s,k,K)$-two layer system such that the following holds:
\begin{enumerate}
\item The ground graph in of $X$ is a $\lambda_\ground$-expander with
$$\lambda_\ground = \frac{(1-\alpha)(1-(s-1)\delta)}{4s(s-1)^2 k(k-1)} \min \left\lbrace \frac{7(1-(s-1) \delta)}{4(1+15(s-1) \delta)}, \frac{1}{2} \right\rbrace .$$
\item All the vertices of $X$ are $\lambda_\loc$-expanders with
$$\lambda_\loc = \frac{(1-\alpha)(1-(s-1)\delta)}{8k (s-1)}.$$
\item The non-intersecting graph is either totally disconnected or a $\lambda_{\nint}$-expander with
$$\lambda_\nint = \frac{R_\nint (1- \alpha) (1 - (s-1)\delta) }{4K}.$$
\end{enumerate}
Denote
$\varepsilon_0 =  \min \lbrace \frac{R_\nint^2 (1- \alpha) (1 - (s-1)\delta)}{4} \frac{1}{K}, \frac{7(1-\alpha)^2 (1-(s-1)\delta)^3}{64(1+15(s-1) \delta) s^3 (s-1)^4} \frac{1}{k^2} \rbrace$. For every $A \subseteq E$, if $A$ is non-empty, $(\delta, \alpha)$-locally small and
$\frac{w(A)}{w(E)} <  \varepsilon_0,$
then there exists $\sigma \in T$ that contains exactly one element of $A$, i.e.,
$$w (\lbrace \sigma \in T : \vert \sigma \cap A \vert =1 \rbrace ) >0,$$
i.e., $X$ has the $((\delta ,\alpha), \varepsilon_0)$-unique neighbor expansion property.
\end{theorem}

\begin{proof}
Let $A \subseteq E$ be a non-empty set such that $A$ is $(\delta, \alpha)$-locally small and
$$\frac{w(A)}{w(E)} <  \varepsilon_0.$$

Define the following sets:
$$D = \lbrace \sigma \in T : \sigma \cap A \neq \emptyset \rbrace,$$
$$D_{\nint}^{\geq 2} = \lbrace \sigma \in T : \exists \tau, \tau ' \in A \cap \sigma, \tau \leftrightarrow \tau ' \rbrace,$$
$$D_\loc = D \setminus D_{\nint}^{\geq 2},$$
$$D_{\loc}^i = \lbrace \sigma \in D_\loc : \vert \sigma \cap A \vert =i, \forall \tau_1, \tau_2 \in \sigma \cap A,  \tau_{1} \cap \tau_{2} \neq \emptyset \rbrace,$$
$$D^1 = \lbrace \sigma \in T : \vert \sigma \cap A \vert =1 \rbrace.$$

We note that by this notation
$$D = D^1 \dot\cup D_{\nint}^{\geq 2} \dot\cup \bigcup_{i=2}^{K}  D_{\loc}^i,$$
and by the definition of the weight function $w$, we have that
$$w(A) \leq w (D^1) + K w (D_{\nint}^{\geq 2}) + \sum_{i=2}^{K} i w (D_{\loc}^i).$$
Therefore, in order to prove the Theorem, it is sufficient to prove that
\begin{equation}
\label{m-ineq1}
w(A) > K w (D_{\nint}^{\geq 2}) + \sum_{i=2}^{K} i w( D_{\loc}^i ).
\end{equation}

We will start by showing that
$$K w( D_{\nint}^{\geq 2} ) < \frac{(1- \alpha) (1 - (s-1)\delta)}{4}  w(A).$$
If the non-intersecting graph is totally disconnected, then $w (D_{\nint}^{\geq 2} ) = 0$ and this holds trivially. Otherwise, by Lemma \ref{non-int graph lemma}, we have that
$$\frac{K}{2 R_\nint}  \left(\lambda_\nint + \frac{1}{R_\nint} \frac{w (A)}{w (E)} \right) \geq \frac{K w (D_{\nint}^{\geq 2})}{w (A)}.$$
By our assumptions, $\lambda_\nint =  \frac{R_\nint (1- \alpha) (1 - (s-1)\delta) }{4K}$ and
$$\frac{w (A)}{w (E)} < \frac{R_\nint^2 (1- \alpha) (1 - (s-1)\delta)}{4K}$$
and therefore:
\begin{dmath*}
\frac{(1- \alpha) (1 - (s-1)\delta)}{4} = \frac{K}{2 R_\nint}  \left( \frac{R_\nint(1- \alpha) (1 - (s-1)\delta)}{4K} + \frac{R_\nint(1- \alpha) (1 - (s-1)\delta)}{4K}  \right)  >
\frac{K}{2 R_\nint}  \left(\lambda_\nint + \frac{1}{R_\nint} \frac{w (A)}{w (E)} \right) \geq \frac{K w (D_{\nint}^{\geq 2})}{w (A)},
\end{dmath*}
i.e.,
$$K w (D_{\nint}^{\geq 2} ) < \frac{(1- \alpha) (1 - (s-1)\delta)}{4} w(A).$$

Thus, in order to prove \eqref{m-ineq1} and complete the proof, it is enough to prove that
$$(1-\frac{(1-\alpha)(1- (s-1)\delta)}{4}) w(A) > \sum_{i=2}^{K} i w (D_{\loc}^i )$$
or equivalently that
\begin{equation}
\label{m-ineq2}
(\alpha + \frac{(1-\alpha)(3 + (s-1)\delta)}{4}) w(A) > \sum_{i=2}^{K} i w (D_{\loc}^i ).
\end{equation}

We note that
\begin{dmath*}
\sum_{v \in V} m_v (A_v, A_v) =
\sum_{v \in V} \sum_{\tau \in A_v} \sum_{\tau ' \in A_v, \tau ' \neq \tau} \sum_{\sigma \in T, \tau, \tau ' \subseteq \sigma} w(\sigma) =
\sum_{\sigma \in T} w(\sigma) \sum_{v \in \sigma} \sum_{\tau \in A_v, \tau \subseteq \sigma} \sum_{\tau ' \in A_v, \tau ' \neq \tau, \tau ' \subseteq \sigma} 1 =
\sum_{\sigma \in T} w(\sigma) \sum_{\tau \in A \cap \sigma} \sum_{\tau ' \in A \cap \sigma, \vert \tau \cap \tau ' \vert =1} 1
 \geq
\sum_{i=2}^{K}
\sum_{\sigma \in D_{\loc}^i} w(\sigma) \sum_{\tau \in A \cap \sigma} \sum_{\tau ' \in A \cap \sigma, \vert \tau \cap \tau ' \vert =1} 1 =
\sum_{i=2}^{K} i (i-1) w (D_{\loc}^i ) \geq
\sum_{i=2}^{K} i  w (D_{\loc}^i) .
\end{dmath*}
Thus, in order to prove \eqref{m-ineq2}, it is sufficient to prove
\begin{equation}
\label{m-ineq2.9}
\sum_{v \in V} m_v (A_v, A_v) < (\alpha + \frac{(1-\alpha)(3 + (s-1)\delta)}{4}) w(A).
\end{equation}

We observe that
\begin{dmath*}
\sum_{v \in V} m_v (A_v, A_v) = \sum_{v \in V_{\delta \text{-small}}} m_v (A_v, A_v) + \sum_{v \in V_{\delta \text{-large}}} m_v (A_v, A_v) \leq
\sum_{v \in V_{\delta \text{-small}}} m_v (A_v, A_v) + \sum_{v \in V_{\delta \text{-large}}} m_v (A_v) \leq^{A \text{ is } (\delta,\alpha)-\text{locally small}}
\sum_{v \in V_{\delta \text{-small}}} m_v (A_v, A_v) + \alpha w(A)
\end{dmath*}
Thus, in order to prove \eqref{m-ineq2.9}, it is enough to prove that

\begin{equation}
\label{m-ineq3}
\sum_{v \in V_{\delta \text{-small}}} m_v (A_v, A_v) <  \frac{(1-\alpha)(3 + (s-1)\delta)}{4} w(A).
\end{equation}

Let $\mu = \frac{(1-\alpha)(1-(s-1)\delta)}{8k (s-1)}$. Applying Lemma \ref{links lemma} on the set $V_{\mu \text{-small}}$ yields
\begin{dmath*}
(s-1) \left(\lambda_\loc + \mu  \right)  k w (A) \geq
(s-1) \left(\lambda_\loc + \mu   \right)k   \sum_{i=0}^k  w (A_{V_{\mu \text{-small}}}^i) \geq
(s-1) \left(\lambda_\loc + \mu   \right)  \sum_{i=1}^k i w (A_{V_{\mu \text{-small}}}^i) \geq^{\text{Lemma } \ref{links lemma}} \sum_{v \in V_{\mu \text{-small}}} m_v (A_v,A_v).
\end{dmath*}
Using the fact that $\mu = \frac{(1-\alpha)(1-(s-1)\delta)}{8k (s-1)}$ and that $\lambda_{\loc} = \frac{(1-\alpha)(1-(s-1)\delta)}{8k (s-1)}$, we get that
\begin{equation}
\label{m-ineq4}
\frac{(1-\alpha)(1-(s-1)\delta)}{4} w (A) \geq \sum_{v \in V_{\mu \text{-small}}} m_v (A_v,A_v).
\end{equation}

Note that since $V_{\delta \text{-small}} \subseteq  V_{\mu \text{-small}} \dot\cup (V_{\mu \text{-large}} \cap V_{\delta \text{-small}})$, we have that
$$\sum_{v \in V_{\delta \text{-small}}} m_v (A_v,A_v) \leq \sum_{v \in V_{\mu \text{-small}}} m_v (A_v,A_v) + \sum_{v \in V_{\mu \text{-large}} \cap V_{\delta \text{-small}}} m_v (A_v,A_v),$$
and thus, in order to prove \eqref{m-ineq3}, we are left to show that

\begin{dmath}
\label{m-ineq5}
\sum_{v \in V_{\mu \text{-large}} \cap V_{\delta \text{-small}}} m_v (A_v,A_v) <  \frac{(1-\alpha)(3+(s-1)\delta)}{4} w(A) - \frac{(1- \alpha)(1-(s-1)\delta)}{4} w(A) = (1-\alpha)\frac{1+(s-1)\delta}{2} w(A).
\end{dmath}

Lemma \ref{links lemma} applied to $V_{\mu \text{-large}} \cap V_{\delta \text{-small}} \subseteq V_{\delta \text{-small}}$ yields that
\begin{dmath*}
(s-1)\left(\lambda_\loc + \delta  \right)  \sum_{i=1}^k i w (A_{V_{\mu \text{-large}} \cap V_{\delta \text{-small}}}^i) \geq \sum_{v \in V_{\mu \text{-large}} \cap V_{\delta \text{-small}}} m_v (A_v,A_v).
\end{dmath*}
Recall that $\lambda_\loc = \frac{(1-\alpha)(1-(s-1)\delta)}{8k (s-1)} \leq\frac{(1-\alpha)(1-(s-1)\delta)}{16 (s-1)}$ and thus we have that
\begin{dmath*}
\left(\frac{(1-\alpha)(1-(s-1)\delta)}{16} + (s-1)\delta  \right)  \sum_{i=1}^k i w (A_{V_{\mu \text{-large}} \cap V_{\delta \text{-small}}}^i) \geq \sum_{v \in V_{\mu \text{-large}} \cap V_{\delta \text{-small}}} m_v (A_v,A_v).
\end{dmath*}
Note that
$$w(A) \geq \sum_{i=1}^k  w (A_{V_{\mu \text{-large}}}^i),$$
and combining this with the above inequality yields
\begin{dmath}
\label{m-ineq5.9}
\left(\frac{(1-\alpha)(1-(s-1)\delta)}{16} + (s-1)\delta  \right)   \left( w(A) + \sum_{i=2}^k (i-1) w (A_{V_{\mu \text{-large}}}^i) \right) \geq \sum_{v \in V_{\mu \text{-large}}} m_v (A_v,A_v),
\end{dmath}
or equivalently,
\begin{dmath}
\label{m-ineq6}
(1-\alpha)\left(\frac{1+ 15(s-1)\delta}{16}  \right)   \left( w(A) + \sum_{i=2}^k (i-1) w (A_{V_{\mu \text{-large}}}^i) \right) \geq \sum_{v \in V_{\mu \text{-large}}} m_v (A_v,A_v).
\end{dmath}
For $\mu =\frac{(1-\alpha)(1-(s-1)\delta)}{8k (s-1)}$, by the condition on $\frac{w(A)}{w(E)}$ we have that $\frac{w(A)}{w(E)} \leq \frac{4 \mu^2}{s^3 (s-1)^2}$ and thus we can apply Lemma \ref{ground graph lemma}
\begin{dmath*}
w(A) \left( \frac{s(s-1)(k-1)}{2 \mu} \lambda_\ground +  \frac{s^3 (s-1)^2}{4 \mu^2} \frac{w(A)}{w(E)} \right)\geq
 \left(1-\frac{s(s-1)(k-1)}{2 \mu} \lambda_\ground  \right) \sum_{i =2}^k  (i-1) w(A_{U}^i).
\end{dmath*}
Recall that $\mu =\frac{(1-\alpha)(1-(s-1)\delta)}{8k (s-1)}$ and thus
\begin{dmath}
\label{m-ineq7}
w(A) \left( \frac{4s(s-1)^2 k(k-1)}{(1-\alpha)(1-(s-1)\delta)} \lambda_\ground +  \frac{16 s^3 (s-1)^4 k^2}{(1-\alpha)^2 (1-(s-1)\delta)^2} \frac{w(A)}{w(E)} \right)\geq
 \left(1-\frac{4s(s-1)^2 k(k-1)}{(1-\alpha)(1-(s-1)\delta)} \lambda_\ground  \right) \sum_{i =2}^k  (i-1) w(A_{U}^i).
\end{dmath}
By our choice of $\lambda_\ground$, we have that
\begin{equation}
\label{m-ineq8}
1-\frac{4s(s-1)^2 k(k-1)}{(1-\alpha)(1-(s-1)\delta)} \lambda_\ground \geq \frac{1}{2},
\end{equation}
and
\begin{equation}
\label{m-ineq9}
\frac{4s(s-1)^2 k(k-1)}{(1-\alpha)(1-(s-1)\delta)} \lambda_\ground   \leq \frac{7(1-(s-1) \delta)}{4(1+15(s-1) \delta)}
\end{equation}
By the fact that
$$\frac{w(A)}{w(E)} < \frac{7(1-\alpha)^2 (1-(s-1)\delta)^3}{64(1+15(s-1) \delta) s^3 (s-1)^4 k^2},$$
we have that
\begin{equation}
\label{m-ineq10}
\frac{16 s^3 (s-1)^4 k^2}{(1-\alpha)^2 (1-(s-1)\delta)^2} \frac{w(A)}{w(E)} < \frac{7(1-(s-1) \delta)}{4(1+15(s-1) \delta)}.
\end{equation}
Combining \eqref{m-ineq7},  \eqref{m-ineq8}, \eqref{m-ineq9}, \eqref{m-ineq10} yields
\begin{dmath*}
w(A) \frac{7(1-(s-1) \delta)}{1+15(s-1) \delta}> \sum_{i =2}^k  (i-1) w(A_{V_{\mu \text{-large}}}^i)
\end{dmath*}
Combining this inequality with \eqref{m-ineq6} yields,
\begin{dmath*}
\sum_{v \in V_{\mu \text{-large}} \cap V_{\delta \text{-small}}} m_v (A_v,A_v) <
(1-\alpha)\left(\frac{1+ 15(s-1)\delta}{16}  \right)  \left( w(A) + w(A) \frac{7(1-(s-1) \delta)}{1+15(s-1) \delta} \right) =
(1-\alpha)\left(\frac{1+ 15(s-1)\delta}{16}  \right)  \left( \frac{1+15(s-1) \delta+7(1-(s-1) \delta)}{1+15(s-1) \delta} \right) w(A)=
(1-\alpha)\frac{1+(s-1) \delta}{2} w(A),
\end{dmath*}
and thus \eqref{m-ineq5} is proven and the proof is complete.
\end{proof}

This Theorem can be rephrased as follows (forgetting the explicit bounds of the expansion):
\begin{theorem}
\label{HDE implies unique neigh exp thm}
Let $s \geq 2, k, K \in \mathbb{N}$ and $0 < R \leq 1$ constants. For every $0 < \delta < \frac{1}{s-1}$ and $0 \leq \alpha <1$ there is a constant $\lambda = \lambda (s,k,K,R,\delta, \alpha) >0$ such that for every $X$ that is $(s,k,K)$-two layer system with $R_\nint \geq R$, if $X$ is $\lambda$-expanding $\HDE$,  then $X$ has the $((\delta ,\alpha), \varepsilon_0)$-unique neighbor expansion property,  with
$$\varepsilon_0 =  \min \lbrace \frac{R^2 (1- \alpha) (1 - (s-1)\delta)}{4} \frac{1}{K}, \frac{7(1-\alpha)^2 (1-(s-1)\delta)^3}{64(1+15(s-1) \delta) s^3 (s-1)^4} \frac{1}{k^2} \rbrace.$$
\end{theorem}

\section{Application of the main expansion Theorem to local testability}
\label{Application of the main expansion Theorem to local testability sec}

Here we combine Theorem \ref{main exp thm - detailed} and Corollary \ref{unique neighbor exp imply amp. loc. test. coro} to deduce amplified local testability for codes modelled over $\HDE$'s

\begin{theorem}
\label{main loc test thm}
Let $p$ be some prime power and fix $1> \delta > \frac{p-1}{p}$. Let $\mathcal{C}_{HDE, \delta, p}$ be the family of all codes $C \subseteq \mathbb{F}_p^V$ such that the following holds:
\begin{itemize}
\item The code $C$ is modelled over a $(2,k,K)$-two layer system $X= X (C)$ and $\frac{R_\nint^2}{K} \geq \frac{1}{k^2}$.
\item The system $X$ is $\lambda $-expanding $\HDE$ with $\lambda (s=2,k,K,R_\nint (X),\delta, \alpha = 0)$ given in Theorem \ref{HDE implies unique neigh exp thm}.
\end{itemize}

Then the family $\mathcal{C}_{HDE, \delta, p}$ is amplified locally testable with $r_{ \mathcal{C}_{HDE, \delta, p}} =  \frac{7 (1-\delta)^3 (\delta - \frac{p-1}{p})}{512(1+15 \delta)}$ and $t_{  \mathcal{C}_{HDE, \delta, p}} = 2$. 
\end{theorem}

\begin{proof}
Applying Theorem \ref{HDE implies unique neigh exp thm} yields that every code $C \in \mathcal{C}_{HDE, \delta, p}$ is modelled over a $(2,k,K)$-two layer system $X$ that has the $((\delta ,0), \varepsilon_0)$-unique neighbor expansion property, with
$$\varepsilon_0 (C) =  \min \lbrace \frac{R_\nint^2 (1 - \delta)}{4} \frac{1}{K}, \frac{7 (1-\delta)^3}{512(1+15 \delta)} \frac{1}{k^2} \rbrace.$$
By the assumption that $\frac{R_\nint^2}{K} \geq \frac{1}{k^2}$, it follows that
$$\varepsilon_0 (C) \geq  \frac{7 (1-\delta)^3}{512(1+15 \delta)} \frac{1}{k^2}.$$
Thus, in the notations of Corollary \ref{unique neighbor exp imply amp. loc. test. coro}, $\mathcal{C}_{HDE, \delta, p} \subseteq \mathcal{C}_{\delta, p, t', \mu}$ with $t' =2$ and $\mu =  \frac{7 (1-\delta)^3}{512(1+15 \delta)}$.
In particular this family is amplified locally testable with $t_{\mathcal{C}_{HDE, \delta, p}} = 3$ and
$$r_{\mathcal{C}_{HDE, \delta, p}} = \frac{7 (1-\delta)^3 (\delta - \frac{p-1}{p})}{512(1+15 \delta)}.$$
\end{proof}

\section{Example: Single orbit affine invariant codes are locally testable}

\label{Example: Single orbit affine invariant codes sec}

Below, we show how to use Theorem \ref{main loc test thm} in order to prove that single orbit affine invariant codes over a field of characteristic $p$ are locally testable. This fact was already proved by the Kaufman and Sudan \cite{KS}, but our proof is different and we show more: we show that the family of single orbit affine invariant codes is amplified locally testable.  

We start with defining single orbit affine invariant codes. Let $n \in \mathbb{N}$ be a positive integer, $p$ be a prime and $\mathbb{K}$ be a finite field. We recall that the following:
\begin{enumerate}
\item The \textit{affine group} $\Aff (\mathbb{K}^n) = \mathbb{K}^n \rtimes \GL (\mathbb{K}^n)$ (where $\GL (\mathbb{K}^n)$ are the invertible linear transformations on $\mathbb{K}^n$) is the group of couples of the form $(\underline{y}, L)$ where $\underline{y} \in \mathbb{K}^n$ and $L \in \GL (\mathbb{K}^n)$ acting on $\mathbb{K}^n$ by
$$(\underline{y}, L). \underline{x} = L \underline{x} + \underline{y}.$$
\item The group $\Aff (\mathbb{K}^n)$ acts on functions $f \in  \mathbb{F}_p^{\mathbb{K}^n}$ by
$$((\underline{y}, L).f) (\underline{x}) = f ((\underline{y}, L)^{-1}.\underline{x}) =f((-L^{-1} \underline{y}, L^{-1}).\underline{x})= f ( L^{-1} \underline{x} - L^{-1} \underline{y})  .$$
\item For $k' \in \mathbb{N}, 1 \leq k' \leq n+1$, a set of $k'$ vectors $\underline{x}_1,...,\underline{x}_{k'}$ are said to be in \textit{general position} if $\underline{x}_2-\underline{x}_1,..., \underline{x}_{k'}-\underline{x}_1$ are linearly independent.
\item For every $k' \in \mathbb{N}, 1 \leq k' \leq n+1$, the group $\Aff (\mathbb{K}^n)$ acts transitively on the set of $k'$-tuples $(\underline{x}_1,...,\underline{x}_{k'})$ of vectors in general position.
\end{enumerate}

A \textit{single orbit affine invariant code} $C \subseteq  \mathbb{F}_p^{\mathbb{K}^n}$ is linear code that is defined as follows: Fix some $\underline{e}_0 \in \mathbb{F}_p^{\mathbb{K}^n}$ such that $\vert \supp (\underline{e}_0) \vert \geq 2$ and such that for every $\underline{x} \in \mathbb{K}^n$ it holds that
$\underline{e}_0 (\underline{x}) \in \lbrace 0,1 \rbrace$. Note that $\underline{e}_0$ is uniquely determined by its support.

Define $C$ to be the linear code defined by
$$\mathcal{E} = \lbrace (\underline{y},L). \underline{e}_0 : (\underline{y},L) \in \Aff (\mathbb{K}^n) \rbrace.$$
In other words, $\underline{c} \in C$ if and only if for every $(\underline{y},L) \in \Aff (\mathbb{K}^n)$ it holds that
$$((\underline{y},L). \underline{e}_0) \cdot \underline{c} = 0.$$
Again, note that every equation in $\mathcal{E}$ is uniquely determined by its support.

We will want to show that any single orbit affine invariant code is modelled over a two layer system $X = (V,E,T)$. By Definition \ref{code modelled over t-l-s def}, $V = \mathbb{K}^n$. Defining $E$ is also straight-forward: Denote $\tau_0 = \supp (\underline{e}_0)$ and note that the code $C$ can be defined via $\tau_0$ as: $\underline{c} \in C$ if and only if for every $(\underline{y},L) \in \Aff (\mathbb{K}^n)$ it holds that
$$\sum_{\underline{x} \in (\underline{y},L).\tau_0} c (\underline{x}) = 0.$$
Thus, we can define $E = \lbrace (\underline{y},L).\tau_0 : (\underline{y},L) \in \Aff (\mathbb{K}^n) \rbrace$. Note the following:
\begin{itemize}
\item For $\Phi$ as in Definition \ref{code modelled over t-l-s def}, the inverse bijection $\Phi^{-1} : E \rightarrow \mathcal{E}$ is
$$(\Phi^{-1} (\tau)) (\underline{x}) = \begin{cases}
1 & \underline{x} \in \tau \\
0 & \underline{x} \notin \tau
\end{cases}.$$
\item The parameter $k$ of $X$ is $k = \vert \tau_0 \vert$ and that all the sets in $E$ are of cardinality $k$.
\end{itemize}
In summation, any choice of a set $\tau_0 \subseteq \mathbb{K}^n$ defines a single orbit affine invariant code as above and two sets $\tau_1, \tau_2$ define the same code $C$ if and only if they are in the same orbit of the action of $\Aff (\mathbb{K}^n)$.

In order to find linear dependencies for a single orbit affine invariant code, we will need some preparation.

Let $\underline{e}_0 \in \mathbb{F}_p^{\mathbb{K}^n}$ be as in the definition of the code $C \subseteq \mathbb{F}_p^{\mathbb{K}^n}$ and denote $\tau_0 = \supp (\underline{e}_0)$ with $\vert \tau_0 \vert =k$. Denote $\tau_0 = \lbrace \underline{x}_1,...,\underline{x}_k \rbrace$ fixing the indices of the vectors in $\tau_0$ (\emph{below everything is done with respect to this fixed indexation}). Call a set $S \subseteq \lbrace 1,...,k \rbrace$ \textit{admissible} if $\lbrace \underline{x}_j : j \in S \rbrace$ is a maximal set of vectors in general position, i.e., if vectors in $\lbrace \underline{x}_j : j \in S \rbrace$ are in general position and for every set $S' \subseteq \lbrace 1,...,k \rbrace$, if $S \subsetneqq S'$, then the vectors in $\lbrace \underline{x}_j : j \in S' \rbrace$ are not in general position. Note that $S$ is admissible if and only if the vectors in $\lbrace \underline{x}_j : j \in S \rbrace$ are in general position and every $\underline{x}_i \in \tau_0$ is in the affine subspace spanned by $\lbrace \underline{x}_j : j \in S \rbrace$, i.e., if for every $j_0 \in S$, $\underline{x}_i - \underline{x}_{j_0} \in \Span \lbrace \underline{x}_j - \underline{x}_{j_0}  : j \in S \rbrace$.

\begin{claim}
There is a constant $2 \leq k' \leq \min \lbrace k, n+1 \rbrace$ such that for every admissible set $S$, $\vert S \vert =k'$.
\end{claim}

\begin{proof}
We will prove that for every two admissible sets $S_1, S_2$ it holds that $\vert S_1 \vert = \vert S_2 \vert$ and thus we can take $k' = \vert S \vert$ where $S$ is any admissible set. Let $S_1, S_2$ be admissible sets. By definition $S_1$ is a maximal set of vectors in general position and thus all the vectors in $S_2$ are in the $(\vert S_1 \vert-1)$-dimensional affine subspace spanned by the vectors in $S_1$. Since the vectors in $S_2$ are in general position, it follows that they span the same affine subspace as the vectors in $S_1$ and thus $\vert S_1 \vert = \vert S_2 \vert$. The fact that $k' \geq 2$ follows from the fact that any two different vectors in $\mathbb{K}^n$ are in general position, the fact that $k' \leq n+1$ follows from the fact that any $(n+2)$ vectors in $\mathbb{K}^n$ are not in general position.
\end{proof}


\begin{proposition}
\label{every 2 indices are in admis prop}
For every $1 \leq i_1, i_2 \leq k$, there is an admissible set $S$ such that $i_1, i_2 \in S$.
\end{proposition}

\begin{proof}
Abusing notation, we say that a set $S \subseteq \lbrace 1,...,k\rbrace$ is in general position if the vectors in $\lbrace \underline{x}_j : j \in S \rbrace$ are in general position. Note that since $i_1 \neq i_2$ implies that $\underline{x}_{i_1} \neq \underline{x}_{i_2}$ it follows that $\lbrace i_1, i_2 \rbrace$ is in general position. Partially order the sets  in general position containing $\lbrace i_1, i_2 \rbrace$ by inclusion and take $S$ to be a set maximal with respect to this ordering.
\end{proof}

\begin{proposition}
\label{linear maps for admis prop}
Let $\underline{e}_0 \in \mathbb{F}_p^{\mathbb{K}^n}$ be as in the definition of the code $C \subseteq \mathbb{F}_p^{\mathbb{K}^n}$ and denote $\tau_0 = \supp (\underline{e}_0)$ with $\vert \tau_0 \vert =k$. Also denote $E = \lbrace (\underline{y},L).\tau_0 : (\underline{y},L) \in \Aff (\mathbb{K}^n) \rbrace$ as above.

For every admissible set $S \subseteq \lbrace 1,...,k\rbrace$ there are linear functions $\ell_1^S,...,\ell_k^S: (\mathbb{K}^n)^{S} \rightarrow \mathbb{K}^n$ such that the following holds:
\begin{itemize}
\item For every $(\underline{y}_j)_{j \in S} \in  (\mathbb{K}^n)^{S}$ and every $i \in S$, $\ell_{i}^S ((\underline{y}_j)_{j \in S}) = \underline{y}_{i}$.
\item For a set  $\tau \subseteq V$ it holds that $\tau \in E$ if and only if there are $\lbrace \underline{y}_j  \in \mathbb{K}^n: j \in S \rbrace$ in general position such that
$$\tau = \lbrace \ell_1^S ((\underline{y}_j)_{j \in S}),...,  \ell_k^S ((\underline{y}_j)_{j \in S}) \rbrace.$$
\end{itemize}
\end{proposition}

\begin{proof}
Denote  $\tau_0 = \lbrace \underline{x}_1,...,\underline{x}_k \rbrace$ as above. Fix $S$ to be an admissible set and fix some $j_0 \in S$. By the definition of admissible sets, for every $1 \leq i \leq k$ there are (unique) scalars $\alpha_{ij} \in \mathbb{K}, j \in S \setminus \lbrace j_0 \rbrace$ such that
$$\underline{x}_i - \underline{x}_{j_0} = \sum_{j \in S \setminus \lbrace j_0 \rbrace} \alpha_{ij} (\underline{x}_j - \underline{x}_{j_0}),$$
or equivalently
$$\underline{x}_i =   \underline{x}_{j_0} + \sum_{j \in S \setminus \lbrace j_0 \rbrace} \alpha_{ij} (\underline{x}_j - \underline{x}_{j_0}).$$
Define $\ell_i^S : (\mathbb{K}^n)^{S} \rightarrow \mathbb{K}^n$ by
\begin{equation}
\label{l_i^S eq}
\ell_i^S ((\underline{y}_j)_{j \in S}) =  \underline{y}_{j_0} + \sum_{j \in S \setminus \lbrace j_0 \rbrace} \alpha_{ij} (\underline{y}_j - \underline{y}_{j_0}) .
\end{equation}

Note that it holds that for every $i \in S$,
$$\ell_i^S ((\underline{y}_j)_{j \in S}) = \underline{y}_i.$$
We are left to prove that $\tau \in E$ if and only if there are $\lbrace \underline{y}_j  \in \mathbb{K}^n: j \in S \rbrace$ in general position such that
$$\tau = \lbrace \ell_1^S ((\underline{y}_j)_{j \in S}),...,  \ell_k^S ((\underline{y}_j)_{j \in S}) \rbrace.$$

First, assume that $\tau = \lbrace \underline{y}_1,...,\underline{y}_k \rbrace \in E$, then (after re-indexing if needed) there is $(\underline{y},L) \in \Aff (\mathbb{K}^n)$ such that for every $1 \leq i \leq k$, $(\underline{y},L).\underline{x}_i = \underline{y}_i$. Thus, for every $1 \leq i \leq k$,
\begin{dmath}
\label{tau comp}
\ell_i^S ((\underline{y}_j)_{j \in S}) = \underline{y}_{j_0} + \sum_{j \in S \setminus \lbrace j_0 \rbrace} \alpha_{ij} (\underline{y}_j - \underline{y}_{j_0}) =
(\underline{y},L). \underline{x}_{j_0} + \sum_{j \in S \setminus \lbrace j_0 \rbrace} \alpha_{ij} ((\underline{y},L).\underline{x}_j - (\underline{y},L).\underline{x}_{j_0}) =
L \underline{x}_{j_0} +\underline{y} +  \sum_{j \in S \setminus \lbrace j_0 \rbrace} \alpha_{ij} ((L\underline{x}_j + \underline{y}) - (L\underline{x}_{j_0}+ \underline{y})) =
\underline{y}  + L(\underline{x}_{j_0} +  \sum_{j \in S \setminus \lbrace j_0 \rbrace} \alpha_{ij} (\underline{x}_j -\underline{x}_{j_0}) =
\underline{y} + L \ell_i^S ((\underline{x}_j)_{j \in S}) =
\underline{y} + L \underline{x}_i = (\underline{y},L). \underline{x}_{i} = \underline{y}_i.
\end{dmath}
Thus,
$$\tau = \lbrace \ell_1^S ((\underline{y}_j)_{j \in S}),...,  \ell_k^S ((\underline{y}_j)_{j \in S}) \rbrace,$$
as needed.

Conversely, let $(\underline{y}_j)_{j \in S}$ be vectors in general position. We will show that
$$\lbrace \ell_1^S ((\underline{y}_j)_{j \in S}),...,  \ell_k^S ((\underline{y}_j)_{j \in S}) \rbrace \in E.$$
Indeed, there is $(\underline{y},L) \in \Aff (\mathbb{K}^n)$ such that for every $j \in S$, $(\underline{y},L).\underline{x}_i = \underline{y}_i$, thus by \eqref{tau comp}, for every $1 \leq i \leq k$,
$$\ell_i^S ((\underline{y}_j)_{j \in S}) = (\underline{y},L). \underline{x}_{i}$$
and it follows that
$$\lbrace \ell_1^S ((\underline{y}_j)_{j \in S}),...,  \ell_k^S ((\underline{y}_j)_{j \in S}) \rbrace \in E.$$
\end{proof}

For an admissible set $S$ denote  $\ell^S: (\mathbb{K}^n)^{S} \rightarrow (\mathbb{K}^n)^k$ to by the linear transformation
$$\ell^S ((\underline{y}_j)_{j \in S}) = (\ell^S_1 ((\underline{y}_j)_{j \in S}),...,\ell^S_k ((\underline{y}_j)_{j \in S})).$$

Denote $M_S$ to be the $k \times k '$ matrix representing $\ell^S$. Explicitly, by equation \eqref{l_i^S eq}, we have that
$$\ell_i^S ((\underline{y}_j)_{j \in S}) =  (1- \sum_{j \in S \setminus \lbrace j_0 \rbrace} \alpha_{ij}) \underline{y}_{j_0} + \sum_{j \in S \setminus \lbrace j_0 \rbrace} \alpha_{ij} \underline{y}_j .$$
Thus, if we denote
$$\beta{ij} = \begin{cases}
\alpha_{ij} & j \neq j_0 \\
(1- \sum_{j \in S \setminus \lbrace j_0 \rbrace} \alpha_{ij}) & j = j_0
\end{cases},$$
where $1 \leq i \leq k$ and $j \in S$, we have that $M_S = (\beta_{ij})_{1 \leq i \leq k, j \in S}$. From the above Proposition it follows that if we treat $\underline{y}_S = (\underline{y}_j)_{j \in S}$ as a column vector, then if $(\underline{y}_j)_{j \in S}$ is in general position, then the support of column vector $M_S (\underline{y}_j)_{j \in S}$ is in $E$. Vice versa, if $\lbrace \underline{y}_i : 1 \leq i \leq k \rbrace \in E$, then (up to re-indexing the $\underline{y}_i$'s) is holds for any admissible set $S$, $M_S (\underline{y}_j)_{j \in S} = (\underline{y}_i)_{1 \leq i \leq k}$ (where $(\underline{y}_i)_{1 \leq i \leq k}$ is a column vector).

Next, we define the concept of general position of matrices.

\begin{definition}[General position for $k' \times k'$ matrix]
Let $S_1, S_2$ be two admissible sets and let $B = (\underline{y}_{j_1,j_2})_{j_1 \in S_1, j_2 \in S_2}$ be a $k' \times k'$ matrix with entries in $\mathbb{K}^n$. We note that $M_{S_1} B$ is a $k \times k '$ matrix, $B M_{S_2}^t$ is a $k' \times k$ matrix and $M_{S_1} B M_{S_2}^t$ is a $k \times k$ matrix.
We say that $B$ is in general position with respect to $S_1, S_2$ if:
\begin{enumerate}
\item Every row in $M_{S_1} B$ is composed of $k'$ vectors in $\mathbb{K}^n$ that are in general position.
\item Every column in $B M_{S_2}^t$ is composed of $k'$ vectors in $\mathbb{K}^n$ that are in general position.
\item The entries of $M_{S_1} B M_{S_2}^t$ are pairwise different.
\end{enumerate}

\end{definition}

\begin{proposition}
If $B = (\underline{y}_{j_1,j_2})_{j_1 \in S_1, j_2 \in S_1}$ is in general position, then the support of every row and every column in $M_{S_1} B M_{S_2}^t$ are in $E$.
\end{proposition}

\begin{proof}
By the definition of general position of $B$, every column in $B M_{S_2}^t$ is composed of $k'$ vectors in $\mathbb{K}^n$ that are in general position. Thus, the support of every column of $M_{S_1} (B M_{S_2}^t)$ is in $E$. Also, every row in $M_{S_1} B$ is composed of $k'$ vectors in $\mathbb{K}^n$ that are in general position. Thus every column of $(M_{S_1} B)^t = B^t M_{S_1}^t$ is in general position and as above the support of every column of $M_{S_2} (B^t M_{S_1}^t)$ is in $E$. By transposing, it follows that the support of every row in $M_{S_1} B M_{S_2}^t$ is in $E$.
\end{proof}

After this set-up, we can finally define of linear dependencies: Fix two admissible sets $S_1, S_2$ and for every $B = (\underline{y}_{j_1,j_2})_{j_1 \in S_1, j_2 \in S_2}$ is in general position define $\ld_{B, S_1,S_2}$ to be
$$\ld_{B, S_1,S_2} (\underline{e}) = \begin{cases}
1 & \supp (\underline{e}) \text{ is a support of a row vector in } M_{S_1} B M_{S_2}^t \\
-1 & \supp (\underline{e}) \text{ is a support of a column vector in } M_{S_1} B M_{S_2}^t \\
0 & \text{otherwise}
\end{cases}.$$

This is a linear dependency since every entry of $M_{S_1} B M_{S_2}^t$ is contained in a row vector and in a column vector and thus they cancel each other. More formally, for $v \in V= \mathbb{K}^n$, we denote $v \in M_{S_1} B M_{S_2}^t$ if $v$ is an entry of $M_{S_1} B M_{S_2}^t$. We also denote $\suppRow (M_{S_1} B M_{S_2}^t), \suppCol (M_{S_1} B M_{S_2}^t)$ to denote the supports rows and columns in $M_{S_1} B M_{S_2}^t$. Then for every $\underline{c} \in \mathbb{F}_p^V$ it holds that
\begin{dmath*}
\sum_{\underline{e} \in \mathcal{E}} \ld_{B, S_1,S_2} (\underline{e}) (\underline{e} \cdot \underline{c})  =
\sum_{\underline{e} \in \mathcal{E}} \ld_{B, S_1,S_2} (\underline{e}) (\sum_{v \in \supp (\underline{e})} \underline{c} (v))  =
\sum_{\underline{e} \in \mathcal{E}, \supp (\underline{e}) \in  \suppRow (M_{S_1} B M_{S_2}^t)} (\sum_{v \in \supp (\underline{e})} \underline{c} (v)) -  \sum_{\underline{e} \in \mathcal{E}, \supp (\underline{e}) \in  \suppCol (M_{S_1} B M_{S_2}^t)} (\sum_{v \in \supp (\underline{e})} \underline{c} (v)) =
\sum_{v \in  M_{S_1} B M_{S_2}^t} \underline{c} (v) - \underline{c} (v) = 0.
\end{dmath*}

The following Proposition shows that the linear dependencies do not depend on the choice of $S_1, S_2$.

\begin{proposition}
\label{indep on S prop}
For a $k \times k$ matrix $M$ and two admissible sets $S_1, S_2$, denote $M \vert_{S_1, S_2}$ to be the $k' \times k'$ matrix minor of $M$ with rows indices in $S_1$ and column indices in $S_2$.

Let $S_1, S_2, S_1', S_2 '$ be admissible sets and $B$ be in general position with respect to $S_1, S_2$. Then $(M_{S_1} B M_{S_2}^t) \vert_{S_1', S_2 '}$ is in general position with respect to $S_1 ', S_2 '$ and
$$M_{S_1 '} (M_{S_1} B M_{S_2}^t) \vert_{S_1 ', S_2 '} M_{S_2 '}^t = M_{S_1} B M_{S_2}^t.$$
In particular, for every $B$ be in general position with respect to $S_1, S_2$ there is a unique $B '$ in general position with respect to $S_1 ', S_2 '$ such that $\ld_{B, S_1,S_2} = \ld_{B ' , S_1 ',S_2 '}$.
\end{proposition}

\begin{proof}
If $S$ is an admissible set and $\lbrace \underline{y}_i : 1 \leq i \leq k \rbrace, \lbrace \underline{y}_i ' : 1 \leq i \leq k \rbrace \in E$, then by Proposition \ref{linear maps for admis prop} if for every $i \in S$ it holds that $\underline{y}_i = \underline{y}_i '$, then it follows that for every $1 \leq i \leq k$, $\underline{y}_i = \underline{y}_i '$. In other words, if two vectors in $E$ agree on their admissible set, they agree everywhere. Thus, in order to show that
$$M_{S_1 '} (M_{S_1} B M_{S_2}^t) \vert_{S_1 ', S_2 '} M_{S_2 '}^t = M_{S_1} B M_{S_2}^t,$$
it is sufficient to show that the row vectors of the two matrices agree on $S_2'$ and the column vectors of the matrices agree on $S_1'$, and this also follows readily from Proposition \ref{linear maps for admis prop}.
\end{proof}

Fix some admissible sets $S_1, S_2$ and define
$$T = \lbrace \supp (\ld_{B, S_1, S_2}) : B \text{ is in general position with respect to } S_1, S_2 \rbrace.$$
Also define a weight function $w$ on $T$ by
$$w (\sigma) = \vert \lbrace B :  \supp (\ld_{B, S_1, S_2}) = \sigma \rbrace \vert.$$
By Proposition \ref{indep on S prop}, $T$ or the weight function do not depend on the choice of $S_1, S_2$. We note that it might be that there are no matrices $B$ in general position with respect to $S_1, S_2$.

We observe that the assuming that $T$ is not an empty set, we know that parameters of $X = (V,E,T)$: Namely, $X$ is a $(2,k,2k)$-system with $R_\nint =1$.

Next, we will prove that given that $\vert \mathbb{K} \vert$ or $n$ are large enough, most $k' \times k'$ matrices with entries in $\mathbb{K}^n$ will in fact be in general position:
\begin{lemma}
\label{counting B in g.p. lemma}
Let $S_1, S_2$ be admissible sets. Then
$$\vert  \lbrace B  : B \text{ is in general position with respect to } S_1, S_2 \rbrace \vert \geq (\vert \mathbb{K} \vert^n)^{(k')^2} (1-\frac{k^2}{\vert \mathbb{K} \vert^n}).$$
Thus if $\vert \mathbb{K} \vert^n > k^2$, then $T$ is non-empty.
\end{lemma}

\begin{proof}
Recall that $B$ is in general position if the following holds
\begin{enumerate}
\item Every row in $M_{S_1} B$ is composed of $k'$ vectors in $\mathbb{K}^n$ that are in general position.
\item Every column in $B M_{S_2}^t$ is composed of $k'$ vectors in $\mathbb{K}^n$ that are in general position.
\item The entries of $M_{S_1} B M_{S_2}^t$ are pairwise different.
\end{enumerate}
Denote $\Mat_{k'} (\mathbb{K}^n)$ to be the space of all $k ' \times k '$ matrices with entries in $\mathbb{K}^n$. This is a $n (k')^2$-dimensional vector space with $(\vert \mathbb{K} \vert^{n})^{(k')^2}$ vectors. Note that the set of matrices $B \in \Mat_{k'} (\mathbb{K}^n)$ such that the first row in $M_{S_1} B$ is \emph{not} in general position is a linear subspace in $\Mat_{k '} (\mathbb{K}^n)$ with co-dimension $1$, i.e., there are $(\vert \mathbb{K} \vert^{n})^{(k')^2-1}$ matrices that violates the condition that the first row in row in $M_{S_1} B$ is in general position. Thus there are at most $k' (\vert \mathbb{K} \vert^{n})^{(k')^2-1}$ matrices in $\Mat_{k '} (\mathbb{K}^n)$ that violate the condition that every row in $M_{S_1} B$ is in general position. Similarly, there are at most $k' (\vert \mathbb{K} \vert^{n})^{(k')^2-1}$ matrices in $\Mat_{k '} (\mathbb{K}^n)$ that violate the condition that every column in $B M_{S_2}^t$ is in general position.

For the third condition, let $1 \leq j_1, j_2, i_1, i_2 \leq k$ be some fixed indices such that $j_1 \neq i_1$ or $j_2 \neq i_2$ (or both) note that by Propositions \ref{every 2 indices are in admis prop}, \ref{indep on S prop}, we can assume that $j_1, i_1 \in S_1$ and $j_2, i_2 \in S_2$. Thus it is clear that the condition that the $(j_1, j_2), (i_1, i_2)$ entries of $M_{S_1} B M_{S_2}^t$ are different is also only violated on a subspace of co-dimension $1$.  Thus, there are at most ${k^2 \choose 2} (\vert \mathbb{K} \vert^{n})^{(k')^2-1}$ matrices $B$ in $\Mat_{k '} (\mathbb{K}^n)$ that violate the condition that the entries of $M_{S_1} B M_{S_2}^t$ are pairwise different.

In conclusion, there are at least
$$(\vert \mathbb{K} \vert^{n})^{(k')^2} - (2 k' + {k \choose 2}) (\vert \mathbb{K} \vert^{n})^{(k')^2-1} \geq (\vert \mathbb{K} \vert^n)^{(k')^2} (1-\frac{k^2}{\vert \mathbb{K} \vert^n})$$
matrices in $\Mat_{k '} (\mathbb{K}^n)$ that are in general position.
\end{proof}

\begin{observation}
\label{transitive obs}
Everything defined above is preserved under the action of $\Aff (\mathbb{K}^n)$. Namely, for every $g \in \Aff (\mathbb{K}^n)$ the following holds:
\begin{itemize}
\item For every $\tau \in E$, $g. \tau \in E$.
\item For every $S_1, S_2$ admissible sets and every $B$ in general position with respect to $S_1, S_2$, $g.B$ (where $g$ acts on the entries of $B$) is also in general position with respect to $S_1, S_2$.
\item For every $S_1, S_2$ admissible sets and every $B$ in general position,
$$g. (M_{S_1} B M_{S_2}^t) =  (M_{S_1} (g.B) M_{S_2}^t)$$
(since $M_{S_1}$, $M_{S_2}^t$ both act on the rows/columns by affine maps). Thus, we can define $g. \ld_{B, S_1, S_2} = \ld_{g.B, S_1, S_2}$.
\item For every $\sigma \in T$, $g. \sigma \in T$ and $w(\sigma) = w(g. \sigma)$. It follows that for every $\tau \in E, v \in V$, $w (g. \tau) = w (\tau)$ and since $\Aff (\mathbb{K}^n)$ acts transitively on $E$, the function $w$ is constant on $E$, i.e., for every $\tau_1, \tau_2 \in E$, $w(\tau_1) = w(\tau_2)$. Similarly, the function $w$ is constant on $V$.
\end{itemize}
\end{observation}

Below, we will show that the following:
\begin{theorem}
\label{good exp in the lim thm}
If $\vert \mathbb{K} \vert^n > 4k^2$,  then system $X=(V,E,T)$ defined above is a $(2,k, 2k)$-system with $R_\nint =1$ that is $\frac{4k^2}{\vert \mathbb{K} \vert^n}$-expanding $\HDE$.
\end{theorem}

It follows that for every $0 <\delta < 1$ and every $0 \leq \alpha <1$, if $\vert \mathbb{K} \vert$ or $n$ are large enough, then the system fulfills the conditions of Theorem \ref{main exp thm - detailed} and thus $X$ has the $(\delta ,\alpha)$-unique neighbor expansion property.

The proof of Theorem \ref{good exp in the lim thm} is based on the idea that (up to passing to a weak cover) the ground graph, the non-intersecting graph and the links all tend to be ``almost complete'' graphs in the sense of Lemma \ref{almost complete graph lemma}. We will start with the ground graph, that is in fact complete (given that $\vert \mathbb{K} \vert$ or $n$ are sufficiently large):

\begin{lemma}
\label{affine inv ground graph lemma}
Let $(V,E,T)$ be the system defined above. If $\vert \mathbb{K} \vert^n > 4k^2$, then the ground graph is a complete graph with a constant weight function on its edges and thus a $0$-expander.
\end{lemma}

\begin{proof}
By Lemma \ref{counting B in g.p. lemma} if $\vert \mathbb{K} \vert^n > 4k^2$, then here are $k' \times k'$ matrices $B$ in general position with respect to $S_1, S_2$ and thus $T$ is non-empty and $w$ is well-defined.

Denote as above $\tau_0 = \lbrace \underline{x}_1,...,\underline{x}_k \rbrace$. We note that every two different vectors  $\underline{y}_1, \underline{y}_2 \in \mathbb{K}^n =V$ are in general position. Since $ \Aff (\mathbb{K}^n)$  acts transitively on pairs of vectors in general position, it follows that there is $g \in \Aff (\mathbb{K}^n)$ such that $g. \underline{x}_1 = \underline{y}_1, g. \underline{x}_2 = \underline{y}_2$. Thus,
$$\underline{y}_1, \underline{y}_2 \in g. \tau_0 \in E,$$
i.e., every pair of vertices of $V$ are contained in a set in $E$ and thus the ground graph is a complete graph.

Recall that
$$m_{\ground} (\lbrace v,u \rbrace) = \sum_{\tau \in E, v,u \in \tau} w (\tau)$$
and that by Observation \ref{transitive obs}, the weight function is constant on the set $E$. Thus, for every $g \in  \Aff (\mathbb{K}^n)$, and every different vertices/vectors $u,v \in V = \mathbb{K}^n$
\begin{dmath*}
m_{\ground} (\lbrace g.v, g.u \rbrace) = \sum_{\tau \in E, g.v, g.u \in \tau} w (\tau) =
\sum_{\tau \in g.E, v, u \in  \tau} w (\tau) =
\sum_{\tau \in E, v, u \in  \tau} w (\tau) = m_{\ground} (\lbrace v, u \rbrace),
\end{dmath*}
and since $\Aff (\mathbb{K}^n)$ acts transitively on pairs of vertices is follows that the weight function is constant on the edges of the ground graph. It follows that the ground graph is complete graph with a constant weight function and in particular a $0$-expander.
\end{proof}

For the non-intersecting graph and the links the analysis is more complicated - we show that these graphs are expander by first passing to a (weak) cover and then using variations of the proof of Lemma \ref{counting B in g.p. lemma}.

Our need to work with a cover of the non-intersecting graph (and of the link) comes from the following complication: as in the proof of Lemma \ref{counting B in g.p. lemma}, we want to basically count matrices of the form $M_{S_1} B M_{S_2}^t$. However, our analysis above relied on a fixing of the indexation, but the $\tau$'s in $E$ are considered as sets without regarding the indexation.

For example, consider the case where $\mathbb{K} = \mathbb{F}_2$, $n> 3$ is large and $\tau_0 = \lbrace \underline{x}_1, \underline{x_2}, \underline{x}_3, \underline{x}_4 = \underline{x}_1 + \underline{x}_2 + \underline{x}_3 \rbrace$ where $\underline{x}_1,\underline{x}_2,\underline{x}_3$ are in general position. One can verify that
$$\underline{x}_4 = (\underline{x}_2-\underline{x}_1) + (\underline{x}_2-\underline{x}_1) + \underline{x}_1,$$
thus $\underline{x}_4$ is in the affine subspace spanned by $\underline{x}_1,\underline{x}_2,\underline{x}_3$. Thus for $S = \lbrace 1,2,3 \rbrace$ and the functions of Proposition \ref{linear maps for admis prop} are
$$\ell^S_j ((\underline{y}_i)_{i \in S}) =   \underline{y}_j, j =1,2,3 ,$$
$$\ell^S_4 ((\underline{y}_i)_{i \in S}) = \underline{y}_1 + \underline{y}_2 + \underline{y}_3.$$
We note that taking $\underline{y}_1 = \underline{x}_2, \underline{y}_2 = \underline{x}_3, \underline{y}_3 = \underline{x}_1$ yields  $\ell^S_4 ((\underline{y}_i)_{i \in S}) = \underline{x}_1+ \underline{x}_2 + \underline{x}_3$ and thus we get $\tau_0$, but with different indexation.

As noted above, in order to overcome this technical problem, we pass to covers. Define $\vec{E}$ as follows: For $\tau_0 = \lbrace \underline{x}_1,..., \underline{x}_k \rbrace$ as above let $\vec{\tau}_0$ be the ordered $k$-tuple  $\vec{\tau}_0= ( \underline{x}_1,..., \underline{x}_k)$ and let $\vec{E}$ to be the orbit of $\vec{\tau}_0$ under all the affine maps. Since matrices of the form $M_{S_1} B M_{S_2}^t$ come equipped with an indexation, one can use them to define the equivalents of the non-intersecting graph and links with respect to $\vec{E}$ and these will be the covers that we will use.

Define the ordered non-intersecting graph to be the graph with a vertex set $\vec{E}$ such that $\vec{\tau}, \vec{\tau} '$ are connected by an edge if their underlying sets do not intersect and there is a matrix of the form $M_{S_1} B M_{S_2}^t$ (with $B$ in general position) such that $\vec{\tau}, \vec{\tau} '$ appear in this matrix as rows or as columns and in such case denote $\vec{\tau} \leftrightarrow \vec{\tau} '$. Define a weight $m_{\nint, o} (\lbrace \vec{\tau}, \vec{\tau} ' \rbrace)$ to be the number of matrices $B$ such that this happens. Define $p : \vec{E} \rightarrow E$ by $p ((\underline{y}_1,...,\underline{y}_k)) = \lbrace \underline{y}_1,...,\underline{y}_k \rbrace$. Note that as in Definition \ref{weak covering def}, the ordered non-intersecting graph is a weak cover of the non-intersecting graph.

\begin{proposition}
\label{ordered non-int graph prop}
Let $(V,E,T)$ be the system defined above.  The ordered non-intersecting graph is a $\frac{4 k^2}{\vert \mathbb{K} \vert^n}$-expander.
\end{proposition}

\begin{proof}
By Lemma \ref{almost complete graph lemma}, it is enough to show that there is ,$\beta = \beta (k,\vert \mathbb{K} \vert, n)$ such that for any $\vec{\tau} \in \vec{E}$,
\begin{equation}
\label{non int ineq cond}
m_{\nint, o} (\vec{\tau}) \geq (\max_{\lbrace \vec{\tau}_1, \vec{\tau}_2 \rbrace, \vec{\tau}_1 \leftrightarrow \vec{\tau}_2} m_{\nint, o} (\lbrace \vec{\tau}_1, \vec{\tau}_2 \rbrace))(1-\beta) \vert \vec{E} \vert,
\end{equation}
and it will follow that the ordered non-intersecting graph is a $2 \beta$-expander. We will show this with $\beta = \frac{2k^2}{\vert \mathbb{K} \vert^n}$.

Let $\vec{\tau}_1, \vec{\tau}_2 \in \vec{E}$ such that their underlying sets do not intersect. Let $S$ is an admissible set with $1,2 \in S$ (this assumption is valid due to Proposition \ref{every 2 indices are in admis prop}). Thus a matrix of the form $M_{S} B M_{S}^t$ has $\vec{\tau}_1$ in the first row and $\vec{\tau}_2$ in the second row if and only if $B$ has the restriction of $\vec{\tau}_1$ to $S$ in the first row and the restriction of $\vec{\tau}_2$ to $S$ in the second row. We will bound the number of such $B$'s that are in general position with respect to $S$. The upper bound is easy - such $B$ has $(k')^2 - 2k'$ free entries thus there are at most $(\vert \mathbb{K} \vert^n)^{(k')^2- 2k'}$ such $B$'s. For the lower bound - we argue as in the proof of Lemma \ref{counting B in g.p. lemma}: For $B$ to be in general position, we need to assign vectors in the $(k')^2 - 2k'$ free entries of $B$ such that
\begin{enumerate}
\item Every row in $M_{S_1} B$ is composed of $k'$ vectors in $\mathbb{K}^n$ that are in general position.
\item Every column in $B M_{S_2}^t$ is composed of $k'$ vectors in $\mathbb{K}^n$ that are in general position.
\item The entries of $M_{S_1} B M_{S_2}^t$ are pairwise different.
\end{enumerate}
Each such condition for row/column/two entries fails on a co-dimension $1$ subspace and thus there are at least
$$(\vert \mathbb{K} \vert^n)^{(k')^2- 2k'} - (2k' + {k \choose 2}) (\vert \mathbb{K} \vert^n)^{(k')^2- 2k'-1}$$
matrices in general position that have $\vec{\tau}_1$ and $\vec{\tau}_2$ in the first and second row.  Applying the same argument to every pair of rows and columns, we deduce the following: If $\vert \mathbb{K} \vert$ or $n$ and sufficiently large,  then every $\vec{\tau}_1$ and $\vec{\tau}_2$ that have non-intersecting underlying sets are connected in the ordered non-intersecting graph and
$$2 k(k-1) ((\vert \mathbb{K} \vert^n)^{(k')^2- 2k'} - (2k' + {k \choose 2}) (\vert \mathbb{K} \vert^n)^{(k')^2- 2k'-1})) \leq m_{\nint, o} (\lbrace \vec{\tau}_1, \vec{\tau}_2 \rbrace) \leq 2 k(k-1) (\vert \mathbb{K} \vert^n)^{(k')^2- 2k'}.$$

It follows that
$$\max_{\lbrace \vec{\tau}_1, \vec{\tau}_2 \rbrace, \vec{\tau}_1 \leftrightarrow \vec{\tau}_2} m_{\nint, o} (\lbrace \vec{\tau}_1, \vec{\tau}_2 \rbrace) \leq 2 k(k-1) (\vert \mathbb{K} \vert^n)^{(k')^2- 2k'}$$
and that for every $\vec{\tau} \in \vec{E}$,
\begin{dmath*}
m_{\nint, o} (\vec{\tau}) \geq
(\max_{\lbrace \vec{\tau}_1, \vec{\tau}_2 \rbrace, \vec{\tau}_1 \leftrightarrow \vec{\tau}_2} m_{\nint, o} (\lbrace \vec{\tau}_1, \vec{\tau}_2 \rbrace))(1- (2k' + {k \choose 2}) \frac{1}{\vert \mathbb{K} \vert^n}){ \vert \lbrace \vec{\tau} ' : \vec{\tau} \leftrightarrow \vec{\tau} ' \rbrace \vert } \geq
(\max_{\lbrace \vec{\tau}_1, \vec{\tau}_2 \rbrace, \vec{\tau}_1 \leftrightarrow \vec{\tau}_2} m_{\nint, o} (\lbrace \vec{\tau}_1, \vec{\tau}_2 \rbrace))(1- \frac{k^2}{\vert \mathbb{K} \vert^n}){ \vert \lbrace \vec{\tau} ' : \vec{\tau} \leftrightarrow \vec{\tau} ' \rbrace \vert }
\end{dmath*}

Therefore, in order to prove \eqref{non int ineq cond}, we are left to prove that
$\frac{\vert \lbrace \vec{\tau} ' : \vec{\tau} \leftrightarrow \vec{\tau} ' \rbrace \vert}{\vert \vec{E} \vert} \geq  1- \frac{k^2}{\vert \mathbb{K} \vert^n}$. (using Bernoulli's inequality). We will not do the exact computation, since it is very similar to the computations above: as we have seen above, every two ordered k-tuples that do not intersect are connected in the ordered non-intersecting graph. The intersection in two entries is defined by a co-dimension $1$ equation and thus
$$\frac{\vert \lbrace \vec{\tau} ' : \vec{\tau} \leftrightarrow \vec{\tau} ' \rbrace \vert}{\vert \vec{E} \vert} \geq 1- \frac{k^2}{\vert \mathbb{K} \vert^n}.$$
\end{proof}

\begin{lemma}
\label{affine inv non-int graph lemma}
Let $(V,E,T)$ be the system defined above.  The non-intersecting graph is a $\frac{4k^2}{\vert \mathbb{K} \vert^n}$-expander.
\end{lemma}

\begin{proof}
By Proposition \ref{weak cover cheeger bound prop}, since the ordered non-intersecting graph is a weak cover of the non-intersecting graph it follows that the Cheeger constant of the non-intersecting graph is bounded by the Cheeger constant of the ordered non-intersecting graph and the assertion of the Lemma follows from the previous Proposition.
\end{proof}

The last graphs we need to consider are the links, i.e.,  for any fixed $v=\underline{x} \in \mathbb{K}^n$,  we need to show that the link of $\underline{x}$ is a $\lambda$-expander with $\lambda$ tending to $0$ as $\vert \mathbb{K} \vert$ or $n$ tend to infinity.  Fix some  $\underline{x} \in \mathbb{K}^n$.  As above,  we will consider a weak cover of the link of $\underline{x}$ and prove that it is an expander and then the expansion of the link will follow from Proposition \ref{weak cover cheeger bound prop}.  Define the \emph{ordered link of $\underline{x}$} to be the graph denoted $\vec{X}_{\underline{x}}$ with a vertex set $\vec{E}_{\underline{x}} = \lbrace \vec{\tau} \in \vec{E} : {\underline{x}} \in \vec{\tau} \rbrace$ and two $\vec{\tau}_1, \vec{\tau}_2 \in \vec{E}_{\underline{x}}$ are connected by an edge in $\vec{X}_{\underline{x}}$ if there is a matrix of the form $M_{S_1} B M_{S_2}^t$ (with $B$ in general position) such that $\vec{\tau}, \vec{\tau} '$ appear in this matrix (one of them as a row and the other as a column). Denote the edge set of the ordered link by $\vec{T}_{\underline{x}}$. Define a weight $m_{v, o} (\lbrace \vec{\tau}, \vec{\tau} ' \rbrace)$ to be the number of matrices $B$ such that this happens. Define $p : \vec{E}_{\underline{x}} \rightarrow E_{\underline{x}}$ by $p ((\underline{y}_1,...,\underline{y}_k)) = \lbrace \underline{y}_1,...,\underline{y}_k \rbrace$. Note that as in Definition \ref{weak covering def}, the ordered link of $\underline{x}$ is a weak cover of the link of $\underline{x}$.

\begin{proposition}
\label{ordered link prop}
Let $(V,E,T)$ be the system defined above. For every $v=\underline{x} \in V$, the ordered link of $\underline{x}$ is a $\frac{4k^2}{\vert \mathbb{K} \vert^n}$-expander.
\end{proposition}

\begin{proof}
Fix $v=\underline{x} \in \mathbb{K}^n$. The proof is similar to the proof of Proposition \ref{ordered non-int graph prop} - by Lemma \ref{almost complete graph lemma}, it is enough to show that for $\beta = \frac{2 k^2}{\vert \mathbb{K} \vert^n}$ it holds that
\begin{equation}
\label{olink ineq cond}
m_{v,o} (\vec{\tau}) \geq (\max_{\lbrace \vec{\tau}_1, \vec{\tau}_2 \rbrace \in  \vec{T}_{\underline{x}}} m_{v, o} (\lbrace \vec{\tau}_1, \vec{\tau}_2 \rbrace))(1-\beta) \vert \vec{E}_{\underline{x}} \vert .
\end{equation}

Let $\vec{\tau}_1, \vec{\tau}_2 \in \vec{E}_{\underline{x}}$ such that their underlying sets intersect only at $\underline{x}$. Let $i$ be the index of $\underline{x}$ in $\vec{\tau}_1$, i.e., if
$$\vec{\tau}_1 = (\underline{y}_1,...,\underline{y}_k),$$
then $\underline{y}_i = \underline{x}$. Let $j$ be the index of $\underline{x}$ in $\vec{\tau}_2$.

Let $S$ is an admissible set with $i,j \in S$ (this assumption is valid due to Proposition \ref{every 2 indices are in admis prop}). Thus a matrix of the form $M_{S} B M_{S}^t$ has $\vec{\tau}_1$ in the $j$-th row and $\vec{\tau}_2$ in the $i$-th column if and only if $B$ has the restriction of $\vec{\tau}_1$ to $S$ in the $j$-th row and the restriction of $\vec{\tau}_2$ to $S$ in the $i$-th column. We will bound the number of such $B$'s that are in general position with respect to $S$. The upper bound is easy - such $B$ has $(k'-1)^2$ free entries thus there are at most $(\vert \mathbb{K} \vert^n)^{(k'-1)^2}$ such $B$'s. For the lower bound - we argue as in the proof of Lemma \ref{counting B in g.p. lemma}: For $B$ to be in general position, we need to assign vectors in the $(k'-1)^2$ free entries of $B$ such that
\begin{enumerate}
\item Every row in $M_{S_1} B$ is composed of $k'$ vectors in $\mathbb{K}^n$ that are in general position.
\item Every column in $B M_{S_2}^t$ is composed of $k'$ vectors in $\mathbb{K}^n$ that are in general position.
\item The entries of $M_{S_1} B M_{S_2}^t$ are pairwise different.
\end{enumerate}
Each such condition for row/column/two entries fails on a co-dimension $1$ subspace and thus there are at least
$$(\vert \mathbb{K} \vert^n)^{(k'-1)^2} - (2k' + {k \choose 2}) (\vert \mathbb{K} \vert^n)^{(k'-1)^2-1}$$
matrices in general position that have $\vec{\tau}_1$ in the $j$-row and $\vec{\tau}_2$ in the $i$-th column.  The same considerations holds if the matrices. Therefore if $\vert \mathbb{K} \vert$ or $n$ and sufficiently large,  then every $\vec{\tau}_1$ and $\vec{\tau}_2$ that their underlying sets intersect only at $\underline{x}$, there are at least
$$2(\vert \mathbb{K} \vert^n)^{(k'-1)^2} - (2k' + {k \choose 2}) (\vert \mathbb{K} \vert^n)^{(k'-1)^2-1}$$
admissible matrices that connect them in the link and at most $2(\vert \mathbb{K} \vert^n)^{(k'-1)^2}$ such matrices.

It follows that
$$\max_{\lbrace \vec{\tau}_1, \vec{\tau}_2 \rbrace \in  \vec{T}_{\underline{x}}} m_{v, o} (\lbrace \vec{\tau}_1, \vec{\tau}_2 \rbrace)) \leq 2(\vert \mathbb{K} \vert^n)^{(k'-1)^2}$$
and that for every $\vec{\tau} \in \vec{E}_{\underline{x}}$,
\begin{dmath*}
m_{v, o} (\vec{\tau}) \geq
(\max_{\lbrace \vec{\tau}_1, \vec{\tau}_2 \rbrace \in  \vec{T}_{\underline{x}}} m_{v, o} (\lbrace \vec{\tau}_1, \vec{\tau}_2 \rbrace))(1- (2k' + {k \choose 2}) \frac{1}{\vert \mathbb{K} \vert^n}){ \vert \lbrace \vec{\tau} ' : \lbrace \vec{\tau}, \vec{\tau} ' \rbrace \in  \vec{T}_{\underline{x}}\rbrace \vert } \geq
(\max_{\lbrace \vec{\tau}_1, \vec{\tau}_2 \rbrace \in  \vec{T}_{\underline{x}}} m_{v, o} (\lbrace \vec{\tau}_1, \vec{\tau}_2 \rbrace))(1-  \frac{k^2}{\vert \mathbb{K} \vert^n}){ \vert \lbrace \vec{\tau} ' : \lbrace \vec{\tau}, \vec{\tau} ' \rbrace \in  \vec{T}_{\underline{x}}\rbrace \vert } .
\end{dmath*}

Therefore, in order to prove \eqref{olink ineq cond}, we are left to prove that
$\frac{\vert \lbrace \vec{\tau} ' : \lbrace \vec{\tau}, \vec{\tau} ' \rbrace \in  \vec{T}_{\underline{x}}\rbrace \vert }{\vert \vec{E}_{\underline{x}} \vert} \geq 1-\frac{k^2}{\vert \mathbb{K} \vert^n}$.  The proof of this fact is exactly as in the proof of Proposition \ref{ordered non-int graph prop} and thus omitted.
\end{proof}

\begin{lemma}
\label{affine inv link graph lemma}
Let $(V,E,T)$ be the system defined above. For every $\underline{x} \in V$ the link of $\underline{x}$ is a $\frac{4k^2}{\vert \mathbb{K} \vert^n}$-expander.
\end{lemma}

\begin{proof}
By Proposition \ref{weak cover cheeger bound prop}, since the ordered link of $\underline{x}$ is a weak cover of the link of $\underline{x}$ it follows that the Cheeger constant of the link of $\underline{x}$ is bounded by the Cheeger constant of the ordered link of $\underline{x}$ and the assertion of the Lemma follows from the previous Proposition.
\end{proof}

\begin{proof}[Proof of Theorem \ref{good exp in the lim thm}]
Combine Lemmas \ref{affine inv ground graph lemma}, \ref{affine inv non-int graph lemma}, \ref{affine inv link graph lemma}.
\end{proof}

Next, for a fixed $\delta > \frac{p-1}{p}$, we want to apply Theorem \ref{main exp thm - detailed} and get a condition for unique neighbor expansion:
\begin{theorem}
\label{affine inv - unique exp thm}
Let $X=(V,E,T)$ be the system defined above and let  $1>\delta > \frac{p-1}{p}$ be a constant. If
$$\vert \mathbb{K} \vert^n \geq k^4 \frac{128  (1+15 \delta)}{7(1- \delta)^2},$$
then $X$ has $(\delta, \varepsilon_0)$-unique neighbor expansion with
$$\varepsilon_0=  \frac{7 (1-\delta)^3}{512(1+15 \delta)} \frac{1}{k^2} .$$
\end{theorem}

\begin{proof}
As noted above, the parameters of the system $X$ are $s= 2, k ,K=2k, R_\nint =1$. Using these parameters in Theorem \ref{main exp thm - detailed} with $\alpha =0$, we get that for $(\delta, \varepsilon_0)$-unique neighbor expansion, it is sufficient that $X$ is $\lambda$-expanding $\HDE$ with
$$\lambda =  \frac{7(1- \delta)^2}{32(1+15 \delta)k(k-1)}.$$
In Theorem \ref{good exp in the lim thm}, we showed that $X$ is $\frac{4k^2}{\vert \mathbb{K} \vert^n}$-expanding $\HDE$. Thus we need that
$$\frac{4k^2}{\vert \mathbb{K} \vert^n} \leq  \frac{7(1- \delta)^2}{32(1+15 \delta)k(k-1)}$$
and this indeed holds if
$$\vert \mathbb{K} \vert^n \geq k^4 \frac{128  (1+15 \delta)}{7(1- \delta)^2}.$$
Applying Theorem \ref{main exp thm - detailed} with the parameters above also gives the stated $\varepsilon_0$.
\end{proof}

Combining this Theorem with Corollary \ref{unique neighbor exp imply amp. loc. test. coro} yields the following amplified local testability result:
\begin{theorem}
\label{affine inv local test gen. thm}
For every $p$ and every $1>\delta > \frac{p-1}{p}$, let $\mathcal{C}_{\text{affine-inv}, \delta, p}$ be the family of affine invariant codes $C \subseteq \mathbb{F}_p^{(\mathbb{K} (C))^{n(C)}}$ such that
$$\vert \mathbb{K} (C) \vert^{n (C)} \geq (k(C))^4 \frac{128  (1+15 \delta)}{7(1- \delta)^2},$$
where $k(C)$ is the size of the support of the equation defining $C$. Then the family $\mathcal{C}_{\text{affine-inv}, \delta, p}$ is amplified locally testable with $r_{\mathcal{C}_{\text{affine-inv}, \delta, p}} = \frac{7 (1-\delta)^3(\delta - \frac{p-1}{p})}{512(1+15 \delta)}$ and $t_{ \mathcal{C}_{\text{affine-inv}, \delta, p}} =3$.
\end{theorem}

\begin{proof}
Combine Theorem \ref{affine inv - unique exp thm} and Corollary \ref{uni-neigh imply in C lemma}.
\end{proof}

We fix $\delta = \frac{2p-1}{2p}$ and derive the following Corollary:
\begin{corollary}
\label{affine inv local test fixed delta coro}
For every $p$, let $\mathcal{C}_{\text{affine-inv}, p}$ be the family of affine invariant codes $C \subseteq \mathbb{F}_p^{(\mathbb{K} (C))^{n(C)}}$ such that
$$\vert \mathbb{K} (C) \vert^{n (C)} \geq 2048p^2 (k(C))^4  ,$$
where $k(C)$ is the size of the support of the equation defining $C$. Then the family $\mathcal{C}_{\text{affine-inv}, \delta, p}$ is amplified locally testable and in particular, for every $C \in \mathcal{C}_{\text{affine-inv}, p}$ and every $\underline{c} \in \mathbb{F}_p^{\mathbb{K} (C)^{n (C)}}$ it holds that
$$\rej (\underline{c}) \geq k (C) \frac{1}{2^{15} p^4} \min \left\lbrace \min_{\underline{c} ' \in C} \Vert \underline{c} - \underline{c} ' \Vert, \frac{1}{k(C)^3} \right\rbrace.  $$
\end{corollary}

\begin{proof}
Fix $\delta = \frac{2p-1}{2p}$ in Theorem \ref{affine inv local test gen. thm} and make some simplifications.
\end{proof}

\section{Sphere corrections for general HDE codes}

\label{Sphere corrections for general HDE codes}

So far, we showed that when a code $C$ is modelled over and $\HDE$ with $s=2$, local testability can be deduced. We note that the condition the $s=2$ comes from our correction method in Lemma \ref{flip bit lemma} that can correct and p-ary code to be $\delta$-locally small if $\delta > \frac{p-1}{p}$. However, in Theorem \ref{main exp thm - detailed}, $\lambda$-expansion of a $\HDE$ yields the $(\delta, \varepsilon_0)$-unique neighbor property given that $\delta < \frac{1}{s-1}$. Thus if $s>2$,  Theorem \ref{main exp thm - detailed} can not be applied using the method correction of Lemma \ref{flip bit lemma}.

In this Section, we define an additional structure of local spheres, and show that "locally" locally testable codes that are modelled over $(s,k,K)$- systems that are HDE's with local spheres are locally testable for general $s \geq 2$.

\subsection{Locally spherical systems}

\begin{definition}[Sphere of a vertex]
Let $X$ be a $(s,k,K)$-two layer system.  For a vertex $v \in V$, the sphere of $v$ is the couple $(V_{\sph (v)}, E_{\sph (v)}) \subseteq (V,E)$ defined as follows: For $\tau \in E$ it holds that $\tau \in E_{\sph (v)}$ if the following holds:
\begin{enumerate}
\item $v \notin \tau$.
\item For every $u \in \tau$, there is $\tau ' \in E$ such that $u,v \in \tau '$.
\item There is $\sigma \in T$ such that $v \in \sigma$ and $\tau \in \sigma$.
\end{enumerate}
Define $V_{\sph (v)} = \bigcup_{\tau \in E_{\sph (v)}} \tau$.
\end{definition}

\begin{definition}[Locally spherical systems]
A two layer system $X$ will be called locally spherical if for every $v ,u \in V$ and every $\sigma \in T$ such that $v,u \in \sigma$ there is $\tau \in E_{\sph (v)}$ such that $\tau \in \sigma$ and $u \in \tau$. Note that in particular, if $X$ is locally spherical, then for every $v$, $V_{\sph (v)} = \lbrace u \in V: \exists \tau' \in E, u,v \in \tau ' \rbrace$.
\end{definition}

\begin{definition}[The opposite graph]
The opposite graph of $X$ is the bipartite graph $V_\opp = V \cup E$ and $\lbrace v, \tau \rbrace \in E_\opp$ if $\tau \in E_{\sph (v)}$. The weight function on this graph is defined by
$$m_{\opp} (\lbrace v, \tau \rbrace) = \sum_{\sigma \in T, v \in \sigma, \tau \in \sigma} w(\sigma).$$
Denote $\lambda_{\opp}$ to be the spectral expansion of this graph (note that here we are using spectral expansion and not the more relaxed Cheeger constant expansion defined above).
\end{definition}

\begin{proposition}
\label{m_opp ineq prop}
If $X$ is locally spherical, then for every $v \in V$,
$$\frac{1}{K} m_\opp (v) \leq w(v) \leq  m_\opp (v).$$
Also, it holds for every $\tau \in E$ that
$$\frac{1}{kK} m_\opp (\tau) \leq w(\tau).$$
\end{proposition}

\begin{proof}
By the assumption that $X$ is locally spherical, for every $\sigma \in T$ such that $v \in \sigma$, there is at least one and at most $K$ $\tau \in E_{\sph (v)}$ such that $v \in \sigma$ and $\tau \in \sigma$. It follows that
\begin{dmath*}
w(v) = \sum_{\sigma \in T, v \in \sigma} w(\sigma) \leq
 \sum_{\sigma \in T, v \in \sigma} \sum_{\tau \in E_{\sph (v)}} w(\sigma) =
 \sum_{\tau \in E, \lbrace v, \tau \rbrace \in E_\opp} m_\opp (\lbrace v, \tau \rbrace) = m_\opp (v),
\end{dmath*}
and that
\begin{dmath*}
w(v) = \sum_{\sigma \in T, v \in \sigma} w(\sigma) \geq
 \frac{1}{K} \sum_{\sigma \in T, v \in \sigma} \sum_{\tau \in E_{\sph (v)}} w(\sigma) =
\frac{1}{K} \sum_{\tau \in E, \lbrace v, \tau \rbrace \in E_\opp} m_\opp (\lbrace v, \tau \rbrace) = \frac{1}{K} m_\opp (v).
\end{dmath*}

Next, for $\tau \in E$,  it holds for every $\sigma \in T$ such that $\tau \in \sigma$ that there is at most $kK$ vertices such that $v \in \sigma$ and $\tau \in E_{\sph (v)}$. Thus,
\begin{dmath*}
w(\tau) = \sum_{\sigma \in T, \tau  \in \sigma} w(\sigma) \geq
\frac{1}{kK} \sum_{\sigma \in T, \tau  \in \sigma} w(\sigma) \sum_{v \in V, v \in \sigma, \tau \in E_{\sph (v)}} 1 =
\frac{1}{kK} \sum_{v \in V, \tau \in E_{\sph (v)}} \sum_{\sigma \in T, \tau  \in \sigma, v \in \sigma} w(\sigma) =
\frac{1}{kK} m_{\opp} (\tau).
\end{dmath*}
\end{proof}

\begin{definition}
Let $X = (V,E,T)$ and $A \subseteq E$. For $0 < \mu <1$, we say that $v \in V$ is $\mu$-spherically small with respect to $A$ if
$\frac{m_{\opp} (E_{\sph (v)} \cap A)}{m_\opp (v)} < \mu$. Otherwise, we say that $v \in V$ is $\mu$-spherically large with respect to $A$. When $A$ is clear from context, we will just say that $v$ is $\mu$-spherically small / $\mu$-spherically large.  We will denote $V_{\mu \text{-sph. small}} \subseteq V$ to be the subset of $\mu$-spherically small vertices and $V_{\mu \text{-sph. large}} \subseteq V$ to be the subset of $\mu$-spherically large vertices.
\end{definition}

The next Theorem states that given that $A \subseteq E$ is small and the opposite graph is a good expander, the weight of vertices that are spherically large with respect to $A$ is negligible.

\begin{theorem}
\label{opp exp implies neg.  V-sphere large thm}
Let $X$ be a $(s,k,K)$-two layer system and assume that the opposite graph is connected. For every $A \subseteq E$ and every $0< \delta <1$ it holds that
$$\frac{w(V_{\delta  \text{-sph. large}})}{w(A)} \leq \frac{4 kK}{3 \delta} (\frac{2 \lambda_\opp}{\sqrt{\delta}} + \frac{s K}{\delta} \frac{w(A)}{w(E)}).$$
\end{theorem}

In order to prove this Theorem, we will need the following general Lemmas regarding bipartite graphs:
\begin{lemma}[Sampling in bipartite graphs]
\label{sampling lemma}
Let $(V,E)$ be a connected bipartite graph with $V = V_1 \cup V_2$ and a weight function $m$. Denote $\lambda$ to be the spectral expansion of this graph.  Fix $U \subseteq V_1$ non-empty and denote $N(U) = \lbrace v \in V_2 : \exists u \in U, \lbrace u,v \rbrace \in E \rbrace$. For $0 < \alpha <1$, define $N(U)_{\geq \alpha}$ to be the set of vertices $v \in  N(U)$ such that
$$\left\vert \frac{m (\lbrace \lbrace u,v \rbrace : u \in U \rbrace)}{m(v)} - \frac{m(U)}{m(V_1)} \right\vert \geq \alpha.$$
Then
$$m(N(U)_{\geq \alpha}) \leq \frac{\lambda^2}{\alpha^2} m(U).$$
\end{lemma}

\begin{proof}
Define
$\phi_U : V \rightarrow \mathbb{R}$ as
$$\phi_U (v) = \begin{cases}
1- \frac{m(U)}{m(V_1)} & v \in U \\
- \frac{m(U)}{m(V_1)} & v \in V_1 \setminus U \\
0 & v \in V_2
\end{cases}.$$
Since the graph is bipartite, it follows that if $\phi_U \perp \mathbbm{1}_{V_1},  \mathbbm{1}_{V_2}$ (where $\mathbbm{1}_{V_i}$ is the indicator function on $V_i$), then $\Vert M \phi_U \Vert^2 \leq \lambda^2 \Vert \phi_U \Vert^2$. It is obvious that $\phi_U \perp \mathbbm{1}_{V_2}$ and we leave it to the reader to check that $\phi_U \perp \mathbbm{1}_{V_1}$. We note that
$$\Vert \phi_U \Vert^2 = m(U) \left(1- \frac{m(U)}{m(V_1)} \right)^2 + m(V_1 \setminus U) \left(-\frac{m(U)}{m(V_1)} \right)^2 = m(U) \left(1- \frac{m(U)}{m(V_1)} \right) \leq m(U).$$
For every $v \in V_2$,
\begin{dmath*}
M \phi_U (v) = {\frac{1}{m(v)} \sum_{u \in U, \lbrace u,v \rbrace \in E} m(\lbrace u,v \rbrace) \left(1- \frac{m(U)}{m(V_1)} \right) + \sum_{u \in V_1 \setminus U, \lbrace u,v \rbrace \in E} m(\lbrace u,v \rbrace) \left(- \frac{m(U)}{m(V_1)} \right)} =
\left( \frac{1}{m(v)} \sum_{u \in U, \lbrace u,v \rbrace \in E} m(\lbrace u,v \rbrace) 1 \right) - \frac{m(U)}{m(V_1)} =
\frac{m (\lbrace \lbrace u,v \rbrace : u \in U \rbrace)}{m(v)} - \frac{m(U)}{m(V_1)}.
\end{dmath*}
Thus
$$\Vert M \phi_U \Vert^2 = \sum_{v \in V_2} m(v) \left\vert \frac{m (\lbrace \lbrace u,v \rbrace : u \in U \rbrace)}{m(v)} - \frac{m(U)}{m(V_1)} \right\vert^2.$$
Note that it follows that
$$\Vert M \phi_U \Vert^2 \geq m(N(U)_{\geq \alpha}) \alpha^2.$$
Thus for the inequality $\Vert M \phi_U \Vert^2 \leq \lambda^2 \Vert \phi_U \Vert^2$, we deduce that
$$m(N(U)_{\geq \alpha}) \leq \frac{\lambda^2}{\alpha^2} m(U)$$
as needed.
\end{proof}

\begin{lemma}
\label{N(U) lemma}
Let $(V,E)$ be a connected bipartite graph with $V = V_1 \cup V_2$ and a weight function $m$. Denote $\lambda$ to be the spectral expansion of this graph. For $U \subseteq V_1$ non-empty, denote $N(U) = \lbrace v \in V_2 : \exists u \in U, \lbrace u,v \rbrace \in E \rbrace$. Fix $0<\delta <1$ and let $W \subseteq N(U)$ such that for every $u \in U$,
$$\frac{\sum_{v \in W,  \lbrace u,v  \rbrace  \in E} m(\lbrace v,u \rbrace)}{m(u)} \geq \delta.$$
Then
$$\frac{m(U)}{m (W)} \leq  \frac{4}{3 \delta} (\frac{2 \lambda}{\sqrt{\delta}} + \frac{m(U)}{m(V_1)}) .$$
\end{lemma}

\begin{proof}
We note that
$$\sum_{u \in U} \sum_{v \in W, \lbrace u,v  \rbrace  \in E} m(\lbrace v,u \rbrace) \geq \delta m(U).$$
Fix $\alpha = 2 \frac{\lambda}{\sqrt{\delta}}$, let $N(U)_{\geq \alpha}$ be as in Lemma \ref{sampling lemma} and $N(U)_{< \alpha} = N(U) \setminus N(U)_{\geq \alpha}$.  By Lemma \ref{sampling lemma}, $m(N(U)_{\geq \alpha}) \leq \frac{\delta}{4} m(U)$.  Therefore
\begin{dmath*}
\delta m(U) \leq \sum_{u \in U} \sum_{v \in W, \lbrace u,v  \rbrace  \in E} m(\lbrace v,u \rbrace) =
\sum_{u \in U} \sum_{v \in W\cap N(U)_{< \alpha}, \lbrace u,v  \rbrace  \in E} m(\lbrace v,u \rbrace) + \sum_{u \in U} \sum_{v \in W\cap N(U)_{\geq \alpha}, \lbrace u,v  \rbrace  \in E} m(\lbrace v,u \rbrace) \leq
\sum_{u \in U} \sum_{v \in W \cap N(U)_{< \alpha}, \lbrace u,v  \rbrace  \in E} m(\lbrace v,u \rbrace)  + m(N(U)_{\geq \alpha}) \leq
\sum_{u \in U} \sum_{v \in W \cap N(U)_{< \alpha}, \lbrace u,v  \rbrace  \in E} m(\lbrace v,u \rbrace)  + \frac{\delta}{4} m(U) =
 {\sum_{v \in W\cap N(U)_{< \alpha}} m (\lbrace \lbrace u,v \rbrace : u \in U \rbrace) + \frac{\delta}{4} m(U)}.
\end{dmath*}
It follows that
\begin{equation}
\label{2 lambda ineq}
m(U) \leq \frac{4}{3 \delta} \sum_{v \in W \cap N(U)_{< \alpha}} m (\lbrace \lbrace u,v \rbrace : u \in U \rbrace).
\end{equation}

Note that for every $v \in W \cap N(U)_{< \alpha}$,
$$\frac{m (\lbrace \lbrace u,v \rbrace : u \in U \rbrace)}{m(v)} -  \frac{m(U)}{m(V_1)} < \alpha$$
and thus
$$m (\lbrace \lbrace u,v \rbrace : u \in U \rbrace)< (\alpha + \frac{m(U)}{m(V_1)}) m(v).$$
Summing over all $v \in W \cap N(U)_{< 2 \lambda}$ yields
$$\sum_{v \in W \cap N(U)_{< 2 \lambda}} \sum_{\lbrace v,u \rbrace \in E, u \in U} m(\lbrace v,u \rbrace) \leq (\alpha + \frac{m(U)}{m(V_1)}) m (W\cap N(U)_{< 2 \lambda}).$$
Combining this with \eqref{2 lambda ineq} (recalling that $\alpha = \frac{2 \lambda}{\sqrt{\delta}}$ yields
$$m(U) \leq  \frac{4}{3 \delta} (\frac{2 \lambda}{\sqrt{\delta}} + \frac{m(U)}{m(V_1)}) m (W \cap N(U)_{< \alpha}) \leq   \frac{4}{3 \delta} (\frac{2 \lambda}{\sqrt{\delta}} + \frac{m(U)}{m(V_1)}) m (W)$$
as needed.
\end{proof}

After this, we are ready to prove Theorem \ref{opp exp implies neg.  V-sphere large thm}:
\begin{proof}
For $U = V_{\delta  \text{-sph. large}}$, denote
$$N(U) = \lbrace \tau \in E : \exists u \in U, \tau \in E_{\sph (u)} \rbrace.$$
Also denote $W = N(U) \cap A$ and apply Lemma \ref{N(U) lemma} with respect to the opposite graph (note that the conditions holds by the definition of $V_{\delta  \text{-sph. large}}$. It follows that
\begin{equation}
\label{ineq sph1}
\frac{m_\opp (U)}{m_\opp (W)} \leq  \frac{4}{3 \delta} (\frac{2 \lambda_\opp}{\sqrt{\delta}} + \frac{m_\opp (U)}{m_\opp (V)}) .
\end{equation}
By Proposition \ref{m_opp ineq prop},
\begin{equation}
\label{ineq sph2}
\frac{m_\opp (U)}{m_\opp (W)} \geq \frac{1}{kK} \frac{w (U)}{w (W)} \geq \frac{1}{kK} \frac{w (V_{\delta  \text{-sph. large}})}{ w (A)}.
\end{equation}
Note that for every $v \in U = V_{\delta  \text{-sph. large}}$,  it holds that
$$\delta m_\opp (v) \leq m_{\opp} (E_{\sph (v)} \cap A).$$
Summing over all $v \in U$ yields
$$\delta m_\opp (U) \leq \sum_{v \in U} m_{\opp} (E_{\sph (v)} \cap A) \leq m_{\opp} (A)\leq^{\text{Proposition } \ref{m_opp ineq prop}} Kk w(A).$$
Also,
$$m_\opp (V) \geq^{\text{Proposition } \ref{m_opp ineq prop}} w(V) \geq^{\text{Corollary } \ref{w(V) and w(E) coro}} \frac{k}{s} w (E).$$
Thus
\begin{equation}
\label{ineq sph3}
\frac{m_\opp (U)}{m_\opp (V)} \leq \frac{s K}{\delta} \frac{w(A)}{w(E)}.
\end{equation}
Combining \eqref{ineq sph1},  \eqref{ineq sph2}, \eqref{ineq sph3} yields
$$ \frac{w (V_{\delta  \text{-sph. large}})}{w (A)}\leq  \frac{4 kK}{3 \delta} (\frac{2 \lambda_\opp}{\sqrt{\delta}} + \frac{s K}{\delta} \frac{w(A)}{w(E)})$$
as needed.
\end{proof}

\subsection{Local to global error correction from spheres}

Let $C \subseteq \mathbb{F}_p^V$ be a code modelled over a two layer system $X$ that is locally spherical. For every vertex $v \in V$, define the sphere code $C_{\sph (v)} \subseteq \mathbb{F}_p^{V_{\sph (v)}}$ to be the code defined by
$$\mathcal{E}_{\sph (v)} = \lbrace \underline{e} \in \mathcal{E} : \supp (\underline{e}) \in E_{\sph(v)} \rbrace.$$
Define a weight function $m_{\sph (v)} : E_{\sph (v)} \cup V_{\sph(v)} \rightarrow \mathbb{R}$ by
$$m_{\sph (v)} (\tau ) = \sum_{\sigma \in T,  \tau \in \sigma, v \in \sigma} w (\sigma)$$
and
$$m_{\sph (v)} (u ) = \sum_{\tau \in E_{\sph (v)}, u \in \tau} m_{\sph (v)} (\tau).$$
Use $m_{\sph (v)}$ as the weight function for the code $C_{\sph (v)} \subseteq \mathbb{F}_p^{V_{\sph (v)}}$.

\begin{definition}[Locally spherical LTC]
Let $C \subseteq \mathbb{F}_p^V$ be a code modelled over a two layer system $X$ that is locally spherical.  We say that the code $C$ is locally spherical LTC if the following holds:
\begin{enumerate}
\item (Spheres are LTC) The sphere codes of all the vertices are uniformly locally testable.
\item (Extendibility) For every $\underline{c} \in C_{\sph (v)}$ there is $a \in \mathbb{F}_p$ such that $\underline{c}_{\ext} \in  \mathbb{F}_p^{V_{\sph (v)} \cup \lbrace v \rbrace}$ defined as
$$\underline{c}_\ext (u) = \begin{cases}
\underline{c} (u) & u \in V_{\sph (v)} \\
a & u=v
\end{cases},$$
fulfils all the equations on $\mathcal{E}$ containing $v$, i.e., that for every $\underline{e} \in \mathcal{E}$, if $v\in \supp (\underline{e})$, then $\underline{e} \cdot \underline{c}_\ext =0$ (note that the support of $\underline{e}$ is in  $V_{\sph (v)} \cup \lbrace v \rbrace$ and therefore this equality is meaningful even though formally $\underline{c}_\ext$ is not defined in $V \setminus (V_{\sph (v)} \cup \lbrace v \rbrace)$).
\end{enumerate}
\end{definition}

Our main theorem in this section is the following. 
\begin{theorem}
\label{local sph. LTC + exp implies global LTC}
Let $C \subseteq \mathbb{F}_p^V$ be a code modelled over a $(s,k,K)$ two layer system $X$ that is locally spherical.  Assume that $C$ is locally spherical LTC.There is $\lambda >0$ such that if $X$ is a $\lambda$-expanding $\HDE$ and $\lambda_\opp \leq \lambda$, then $C$ is locally testable.
\end{theorem}

Sketch of proof: For $\underline{c} \in \mathbb{F}_p^V$, we denote
$$\mathcal{A} (\underline{c}) = \lbrace \underline{e} \in \mathcal{E} : \underline{e} \cdot \underline{c} \neq 0 \rbrace,$$
and
$$A(\underline{c}) = \lbrace \tau \in E: \exists \underline{e} \in \mathcal{A} (\underline{c}),  \supp (\underline{e}) = \tau \rbrace.$$
Our objective is to show that if $\frac{w (A)}{w(E)}$ is small, then $\underline{c}$ is close to $\underline{c} ' \in C$. We proceed via the following steps:
\begin{enumerate}
\item We fix $\delta < \frac{1}{s-1}$.
\item By locally spherical LTC, every vertex with a small sphere with respect to $A$ can be corrected such that it has a $\delta$-small link.
\item From the previous step it follows that we have a correction scheme that corrects every vertex with a small sphere to a vertex with a  $\delta$-small link. We bound the weight of the corrected vertices in the scheme (similar to the proof of Theorem \ref{unique neighbor exp imply LTC thm}) and show that there is $\underline{c} '$ close to $\underline{c}$ such that every vertex with a  $\delta$-small sphere with respect to $A(\underline{c} ')$  has a small link with respect to $A(\underline{c} ')$.  We will show that $\underline{c} ' \in C$.
\item For $\underline{c} '$ above, the weight of $A(\underline{c} ')$ summed in the  $\delta$-large links is bounded by the weight of spherically large vertices. Thus, by Theorem \ref{opp exp implies neg.  V-sphere large thm} (using the expansion of the opposite graph) the weight of $\delta$-large links is negligible with respect to $A(\underline{c} ')$.
\item By Theorem \ref{main exp thm - detailed}, the $\lambda$-expansion of $X$ implies that every set $A(\underline{c} ')$ with negligible weight of $\delta$-large links has the unique neighbor property.  It follows by Lemma \ref{uni-neigh imply in C lemma}, that  $\underline{c} ' \in C$ as needed.
\end{enumerate}

\section{Distance of codes modelled over two layer systems}
\label{Distance of codes modelled over two layer systems sec}

Here we show that we can bound the minimal distance of codes modelled over a $(s,k,K)$-two layer system based on the expansion of the ground graph. Namely, we show the following:
\begin{theorem}
\label{distance thm}
Let $X = (V,E,T)$ be a $(s,k,K)$-two layer system and $C \subseteq \mathbb{F}_p^V$ be a linear code modelled over $X$. Then for every $\underline{c} \in C \setminus \lbrace \underline{0} \rbrace$ it holds that
$$\Vert \underline{c} \Vert \geq \frac{16}{s^4 (s-1)^2 k} \left( 1- s(s-1)(k-1)\lambda_\ground \right) .$$
\end{theorem}

\begin{proof}
Let $\underline{c} \in C \setminus \lbrace \underline{0} \rbrace$. Denote
$$U = \supp (\underline{c}) = \lbrace v \in V : \underline{c} (v) \neq 0 \rbrace,$$
then $\Vert \underline{c} \Vert = \frac{w(U)}{w(V)}$.
Define $A \subseteq E$ to be
$$A = \lbrace \tau \in E : U \cap \tau \neq \emptyset \rbrace.$$
Denote
$$A_U^i = \lbrace \tau \in A : \vert U \cap \tau \vert =i \rbrace.$$
Since $\underline{c}$ is a code word it follows that $A_U^1$ is empty: Otherwise, there is an equation that meets the support of $\underline{c}$ at exactly one variable and thus not satisfied. Therefore
$$w (A) = \sum_{i=2}^k w(A_U^i).$$
Thus,
\begin{dmath*}
w(A) =
\sum_{i=2}^k w(A_U^i) =
\sum_{i=2}^k \sum_{v \in U} \frac{1}{i} \sum_{\tau \in A_U^i, v \in \tau} w (\tau) \leq
\frac{1}{2}  \sum_{v \in U} \sum_{i=2}^k \sum_{\tau \in A_U^i, v \in \tau} w (\tau) \leq
\frac{1}{2}  \sum_{v \in U} \sum_{\tau \in A, v \in \tau} w (\tau) \leq
\frac{1}{2} \sum_{v \in U} \sum_{\tau \in E, v \in \tau} w(\tau) \leq^{\text{Proposition } \ref{w(v) as sum of w(tau)'s prop}}
\frac{1}{2} \sum_{v \in U}  s w (v) = \frac{s}{2} w (U).
\end{dmath*}
By Corollary \ref{w(V) and w(E) coro}, $\frac{k}{2} w(E) \geq  w(V)$ and thus
$$\Vert \underline{c} \Vert = \frac{w(U)}{w(V)} \geq \frac{4}{sk} \frac{w(A)}{w(E)}.$$
It follows that in order to prove the assertion stated above, we need to prove that
$$\frac{w(A)}{w(E)} \geq \frac{4}{s^3 (s-1)^2} \left( 1- s(s-1)(k-1)\lambda_\ground \right).$$

If $\frac{w(A)}{w(E)} > \frac{4}{s^3 (s-1)^2}$ we are done, and thus we can assume that $\frac{w(A)}{w(E)} \leq \frac{4}{s^3 (s-1)^2}$. Note that by definition, every $v \in U$ is $1$-large with respect to $A$. Thus, by Lemma \ref{ground graph lemma} (with $\mu = 1$), we get that if $\frac{w (A)}{w(E)} \leq \frac{4}{s^3 (s-1)^2}$, then
\begin{dmath*}
\frac{s(s-1)(k-1)}{2} \lambda_\ground +  \frac{s^3 (s-1)^2}{4} \frac{w(A)}{w(E)} \geq
 \left(1-\frac{s(s-1)(k-1)}{2} \lambda_\ground  \right) \sum_{i =2}^k  (i-1) \frac{w(A_{U}^i)}{w(A)} \geq
 \left(1-\frac{s(s-1)(k-1)}{2} \lambda_\ground  \right) \frac{1}{w(A)} \sum_{i =2}^k  w(A_{U}^i) = 1-\frac{s(s-1)(k-1)}{2} \lambda_\ground .
\end{dmath*}
It follows that
$$\frac{w(A)}{w(E)} \geq \frac{4}{s^3 (s-1)^2} \left( 1- s(s-1)(k-1)\lambda_\ground \right),$$
as needed.
\end{proof}

\appendix

\bibliographystyle{alpha}
\bibliography{bibl}
\end{document}